\documentclass[9pt, nocopyrightspace, natbib, reprint]{sigplanconf}

\usepackage{graphicx}
\graphicspath{{./figs/}}

\usepackage{algorithm}
\usepackage{algorithmic}
\usepackage{amsmath}
\usepackage{amsthm}
\usepackage{amsfonts} % for \checkmark
\usepackage{bm}
\usepackage{subfigure}
\usepackage{booktabs}
\usepackage{multirow}
\usepackage{color}

\makeatother

\begin{document}

\title{
L-Shape Based Layout Fracturing for E-Beam Lithography
}

\iffalse

\authorinfo{}{\vspace{-.2in}}

\else

\authorinfo{Bei Yu \and Jhih-Rong Gao \and  David Z. Pan}
           {ECE Dept. Univ. of Texas at Austin, Austin, TX USA}
           {\{bei, jrgao, dpan\}@cerc.utexas.edu \vspace{-.5in}}
\fi

\maketitle

\newtheorem{problem}{\textbf{Problem}}
\newtheorem{mydefinition}{\textbf{Definition}}
\newtheorem{mytheorem}{\textbf{Theorem}}
\newtheorem{mylemma}{\textbf{Lemma}}
\newtheorem{conjecture}{Conjecture}

\begin{abstract}

Layout fracturing is a fundamental step in mask data preparation and e-beam lithography (EBL) writing.
To increase EBL throughput, recently a new L-shape writing strategy is proposed, which calls for new L-shape fracturing, versus the conventional rectangular fracturing. Meanwhile, during layout fracturing, one must minimize very small/narrow features, also called slivers, due to manufacturability concern.
This paper addresses this new research problem of how to perform L-shaped fracturing with sliver minimization.
We propose two novel algorithms. The first one, rectangular merging (RM), starts from a set of rectangular fractures and merges them optimally to form L-shape fracturing. 
The second algorithm, direct L-shape fracturing (DLF), directly and effectively fractures the input layouts into L-shapes with sliver minimization.
%Our experimental results are very promising for both performance improvement and speed-up.
The experimental results show that our algorithms are very effective.

\end{abstract}

\vspace{-.1in}
\section{Introduction}

E-Beam lithography (EBL) \cite{EBL_SPIE05_Pain} is widely deployed in the mask manufacturing,
which is a significant step affecting the fidelity of the printed image on the wafer and critical dimension (CD) control.
Because of the capability of accurate pattern generation, EBL is also a promising candidates for sub-22nm logic nodes,
along with extreme ultra violet (EUV) \cite{EUV_SPIE2010_Arisawa} and double/multiple patterning lithography (DPL/MPL) \cite{DPL_ICCAD08_Kahng}\cite{TPL_ICCAD2011_Yu}.
%EBL has been used in small volume LSI production and R\&D to develop the technological nodes ahead of mass production.
For EBL writing, a fundamental step is \textit{layout fracturing}, where the layout pattern is decomposed into numerous non-overlapping rectangles.
Subsequently the layout is prepared and exposed by an EBL writing machine onto the mask or the wafer,
where each fractured rectangle is shot by one variable shaped beam (VSB).

% Cost increasing of mask writing
As the minimum feature size further decreases, the number of rectangles in the layout is steadily increased.
First, longer writing time and larger data volume are caused by highly complex optical proximity correction (OPC).
%reticle enhancement technique (RET) that is used to optimize mask apertures for the manufacture of deeply sub-wavelength features.
Besides, the introduction of advanced lithographic techniques, e.g., DPL/MPL, add more masks in the mask manufacturing.
Since the manufacturing cost is directly associated with increasing write time and data volume, the cost is also steadily increased.
In addition, the low throughput has been and is still the bottleneck of EBL writing.
%At sub-22nm logic nodes, the number of shots after fracturing may be too large to be practically written by EBL.

To overcome this manufacturing problem, several optimization methods have been proposed to reduce the EBL writing time to a reasonable level
\cite{EBL_SPIE2010_Sahouria}\cite{EBL_SPIE2011_Elayat}\cite{EBL_TCAD2012_Yuan}.
%These optimization methods contain both hardware improvement and software speed-up
%i.e., jog alignment, multi-resolution writing, character projection, and L-shape shot
Among them, the L-shape shot strategy is a very simple yet effective approach to reduce the e-beam mask writing time,
and thus reduce the mask manufacturing cost and improve the throughput
\cite{EBL_SPIE2010_Sahouria}\cite{EBL_SPIE2011_Elayat}.
Besides, this technique can be also applied to reduce the cost of lithographic process.
The conventional EBL writing is based on rectangular VSB shots.
As illustrated in Fig. \ref{fig:EBL}(a), the electrical gun generates an initial beam, which becomes uniform through the shaping aperture.
Then the second aperture finalizes the target shape with a limited maximum size.
As an improved technique, the printing process of the L-shape shot is illustrated in Fig. \ref{fig:EBL}(b).
One additional aperture, the third aperture, is employed to create L-shape shots.
To take advantage of this new printing process, new fracturing methodology is needed to provide L-shape in the fractured layout.
%With careful design, we can see that the
L-shape shot strategy has the potentiality to reduce the EBL writing time or cost by 50\% if all rectangles are combined into L-shapes.
For example in Fig. \ref{fig:example}, instead of four rectangles,
using L-shape fracturing only requires two L-shape shots.

%Besides, until now there is no known mask writer format or methodology has a mechanism to describe the L-shape based fracturing.
%The limiting factor in the implementation of this technology is the development of a write tool with the additional aperture required to create L-shaped shots and the respective changes to data format and mask writer software to support it.

\begin{figure}[bt]
  \centering
  \subfigure[]{\includegraphics[width=0.23\textwidth]{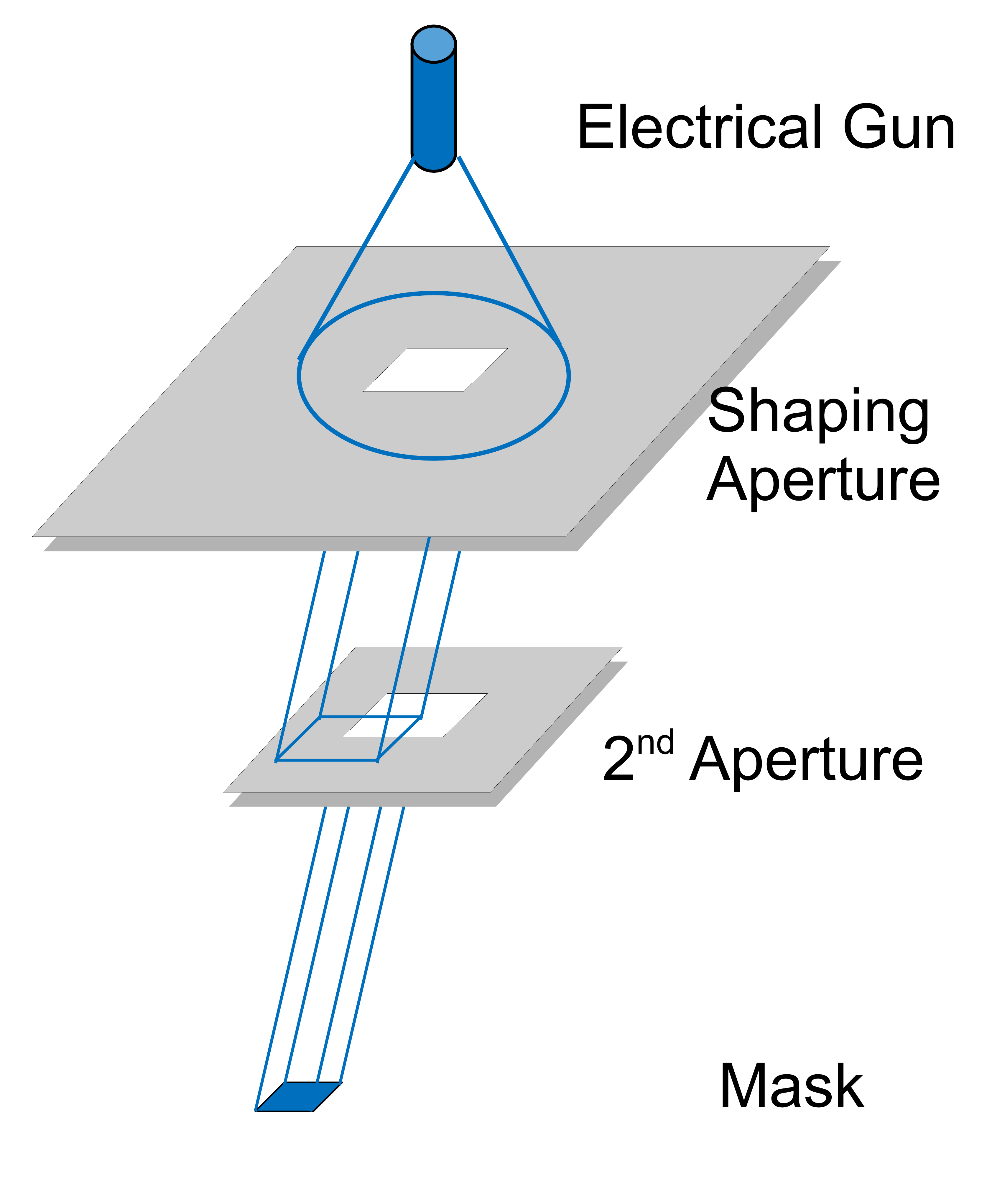}}
  \hspace{-.1in}
  \subfigure[]{\includegraphics[width=0.23\textwidth]{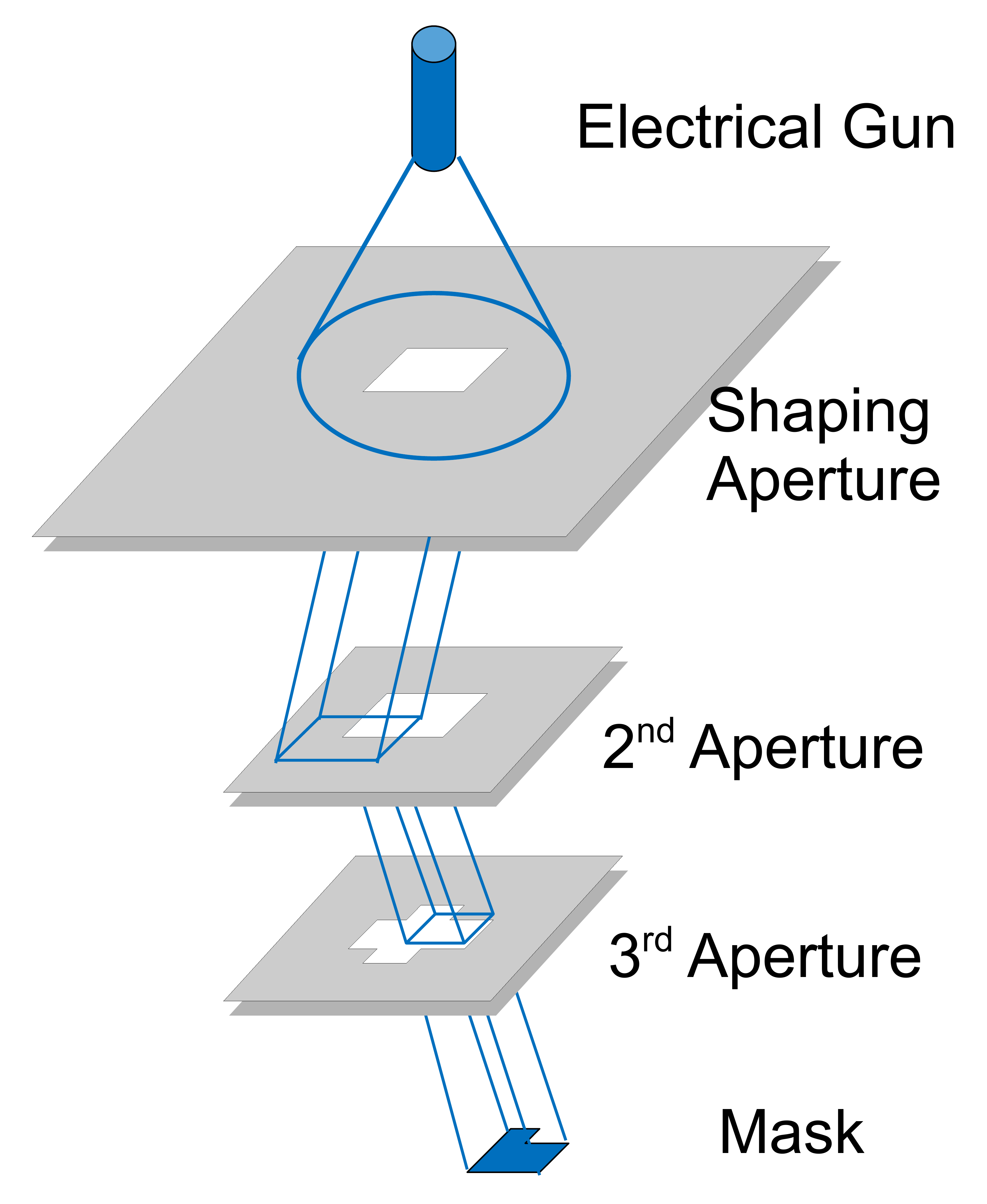}}
  \nocaptionrule
  \caption{~(a) Traditional rectangular EBL writing process.~(b) L-shape writing process with one additional aperture.}
  \label{fig:EBL}
  \vspace{-.1in}
\end{figure}

\begin{figure}[bt]
  \centering
  \subfigure[]{\includegraphics[width=0.14\textwidth]{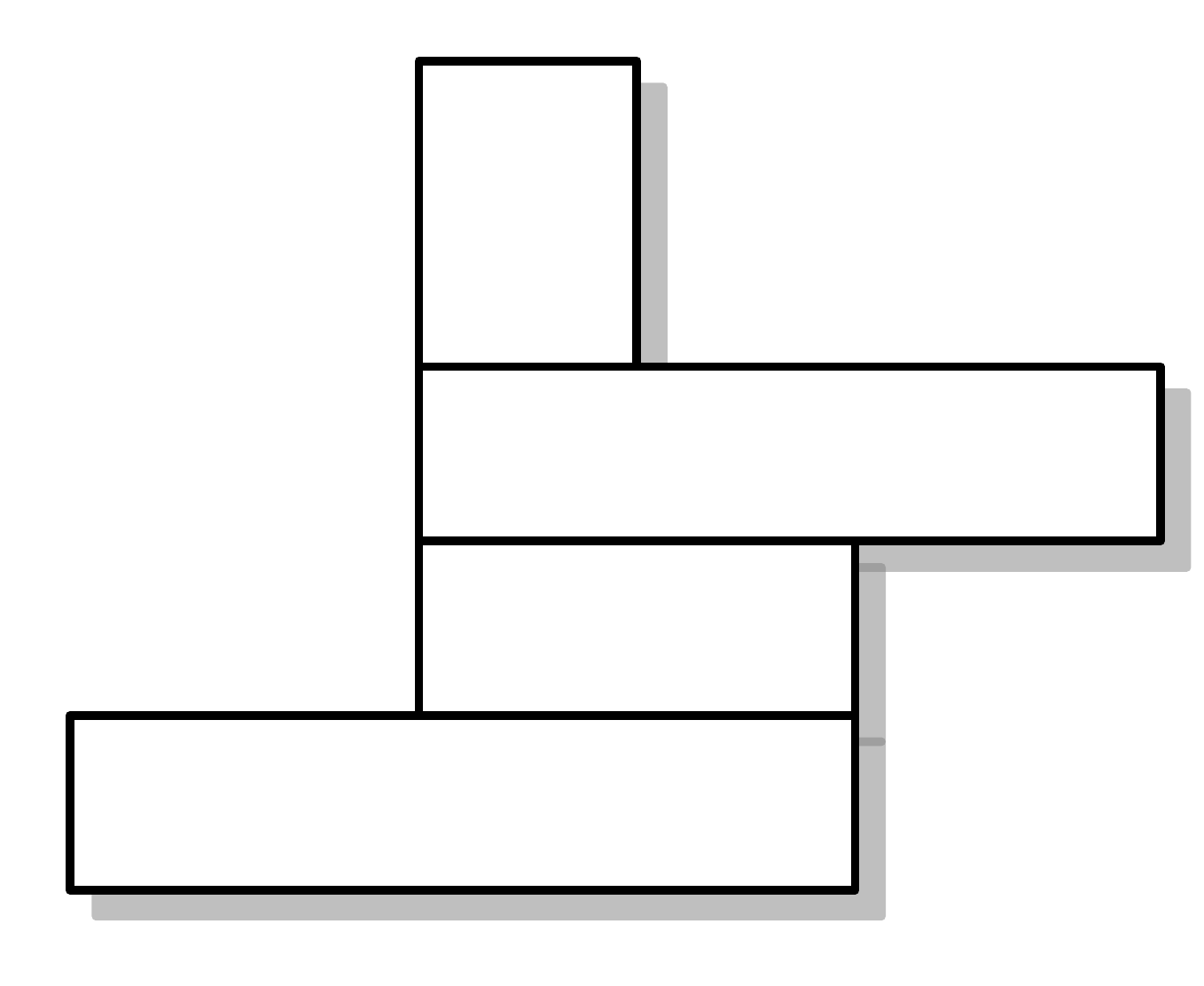}}
  \hspace{.1in}
  \subfigure[]{\includegraphics[width=0.14\textwidth]{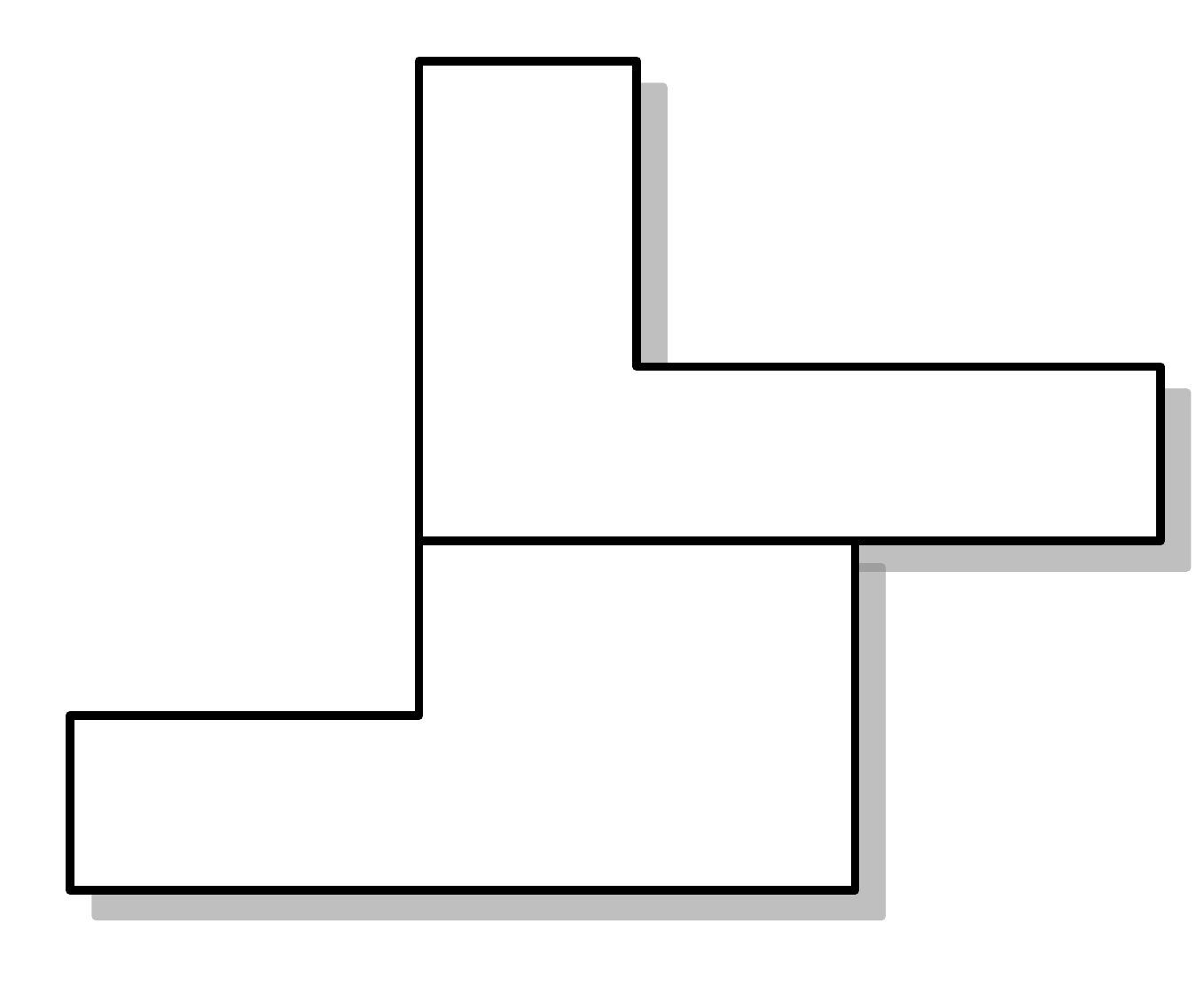}}
  \nocaptionrule
  \caption{Examples of polygon fracturing.~(a) Rectangular shots with 4 shot number.~(b) L-shape shots with 2 shot number.}
  \label{fig:example}
  \vspace{-.1in}
\end{figure}

Note that the layout fracturing problem is different from the general polygon decomposition problem in geometrical science.
In order to consider yield control and CD control,
the minimum width of each shot should be above a certain threshold value $\epsilon$.
A shot whose minimum width is $< \epsilon$ is called a \textit{sliver}.
In the layout fracturing, sliver minimization is an important objective \cite{EBL_SPIE04_Kahng}.
As shown in Fig. \ref{fig:sliver}, two fractured layouts can achieve the same shot number 2.
However, because of sliver, the fractured result in Fig. \ref{fig:sliver} (a) is worse than that in Fig. \ref{fig:sliver} (b).
It shall be noted that the layout in  Fig. \ref{fig:sliver} can be written in one L-shaped shot without any sliver.

\begin{figure}[tb]
  \centering
  \vspace{-.1in}
  \subfigure[]{
    \includegraphics[width=0.18\textwidth]{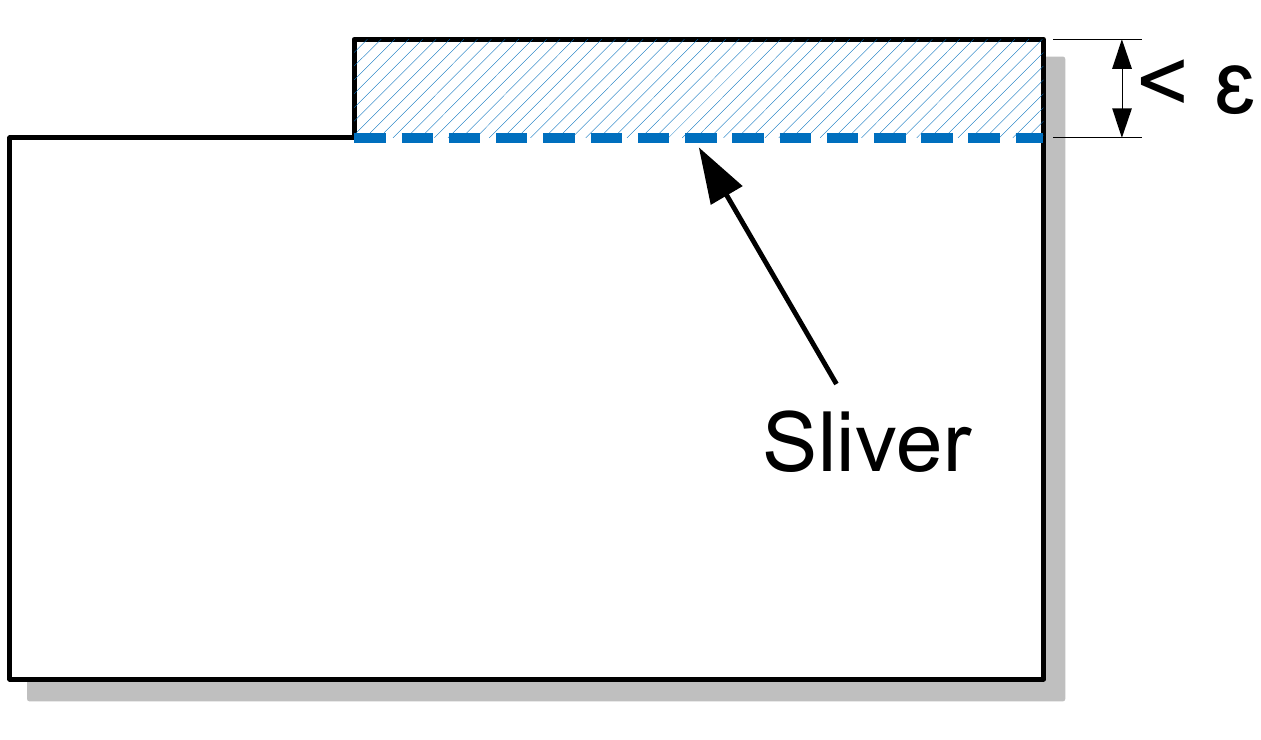}
  }
  %\hspace{-.1in}
  \subfigure[]{
    \includegraphics[width=0.18\textwidth]{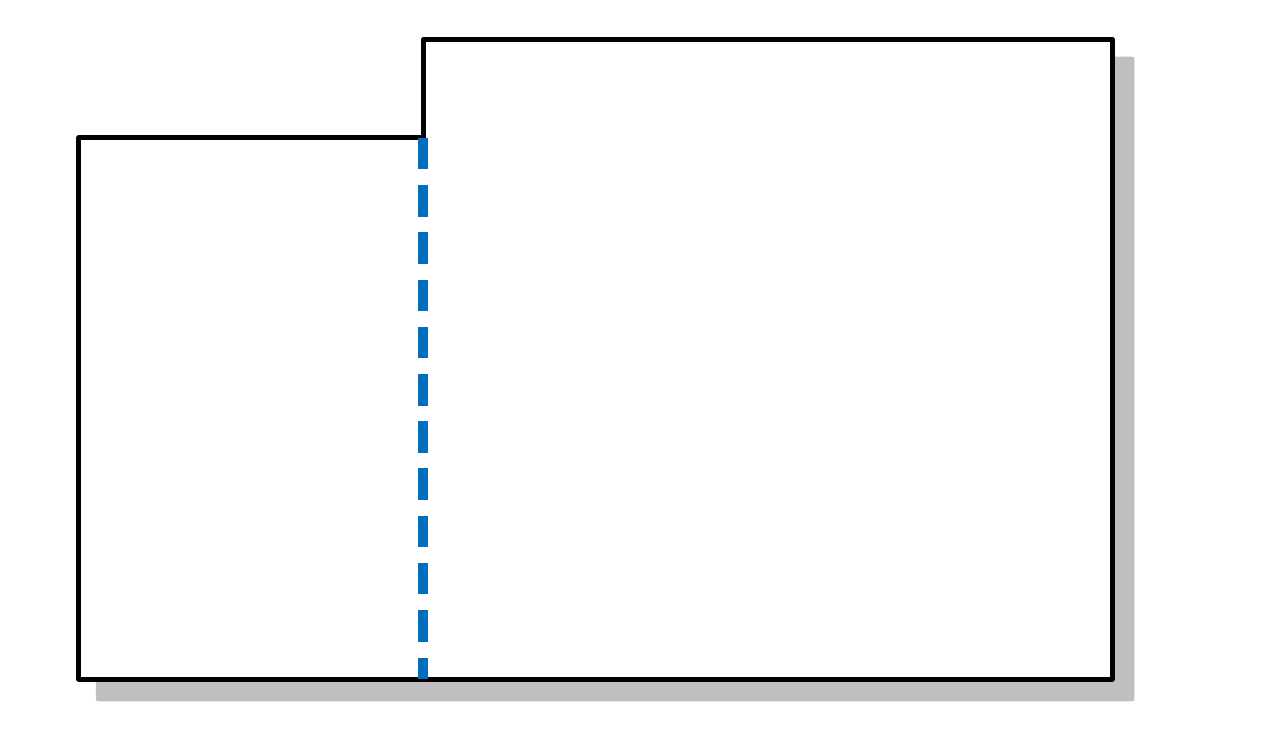}
  }
  \nocaptionrule
  \caption{~(a) Fracturing with one sliver.~(b) Fracturing without sliver.}
  \label{fig:sliver}
  \vspace{-.1in}
\end{figure}

% Previous Works
For traditional rectangular shots, several papers have studied the layout fracturing problem
\cite{EBL_SPIE04_Kahng}\cite{EBL_SPIE06_Kahng}\cite{EBL_SPIE08_Dillon}\cite{EBL_SPIE2011_Ma}\cite{EBL_SPIE2011_Jiang}.
Kahng et. al proposed an integer linear programming (ILP) formulation, and some matching based speed-up techniques
\cite{EBL_SPIE04_Kahng}\cite{EBL_SPIE06_Kahng}.
Recently, Ma et. al \cite{EBL_SPIE2011_Ma} presented a heuristic algorithm to generate rectangular shots and further reduce the sliver.
Compared with the rectangular fracturing problem, the L-shape fracturing problem is new and there is only limited work,
mostly describing methodology, but no systematic algorithm has been proposed so far.
\cite{EBL_SPIE2010_Sahouria} reported the initial results that L-shape fracturing can further save about 38\% of shot count, but no algorithmic details are provided.
For the general decomposition problem of polygon into L-shapes, several heuristic methods are proposed
\cite{ALG_1984_Edelsbrunner}\cite{ALG_TODAES96_Lopez}.
However, since these heuristic methods only consider horizontal decomposition, which would result in numerous slivers, 
they cannot be applied to the layout fracturing problem.

This paper presents the first systematic study for EBL L-shape fracturing considering the sliver minimization.
We propose two algorithms for the L-shape fracturing problem.
The first method, called \textit{RM}, starts from rectangles generated by any previous fracturing framework, 
and merge them into L-shapes.
A maximum weighted matching algorithm is proposed to find the optimal merging solution,
where the shot count and the sliver can be minimized simultaneously.
To further overcome the intrinsic limitations of rectangular merging,
we propose another fracturing algorithm, called \textit{DLF}.
Through effectively detect and take advantage of some special cuts,
DLF can directly fracture the layout into a set of L-shapes in $O(n^2logn)$ time.
The experimental results show that our algorithms are very promising for both shot count reduction and sliver minimization.
In addition, DLF can even achieve significant speed-up compared with previous state-of-the-art rectangular fracturing algorithm \cite{EBL_SPIE2011_Ma}.
%Compared with previous works and the rectangular merging method, DLF can not only effectively reduce the sliver length and shot number,
%but also achieve good speed-up.

\iffalse
In this paper, we present two L-shape fracturing methods.
The first method, called \textit{RM}, is a rectangular merging based incremental method.
In other words, previous rectangular fracturing algorithms can be re-used and extended.
In our second method, we propose a high performance L-shape fracturing algorithm, called \textit{DLF},
which can directly fracture a rectilinear polygon into a set of L-shapes.
Compared with previous works and the RM method, DLF can not only effectively reduce the sliver length and shot number,
but also achieve good speed-up.
Our key contributions are the following:
\begin{itemize}
  \item Propose the first systematic study for L-shape fracturing problem considering the sliver minimization;
  \item Develop a rectangular merging based method to re-use and extend previous fracturing methodology;
  \item Propose an $O(n^2logn)$ algorithm to directly fracture polygons into L-shapes;
  \item Our experimental results are very promising for both performance improvement and speed-up.
\end{itemize}
\fi

The rest of the paper is organized as follows.
Section \ref{sec:problem} presents the basics and problem formulation.
Section \ref{sec:post} provides RM, the merging based algorithm, which will also be used as a baseline.
In Section \ref{sec:algo} we propose the DLF algorithm to directly fracture polygons into L-shapes.
%Section 4 shows how we further integrate the density balance into our decomposer.
Section \ref{sec:result} presents experimental results, followed by conclusion in Section \ref{sec:conclu}.

\vspace{-.1in}
\section{Definitions and Problem Formulation}
\label{sec:problem}

We first introduce some notations and definitions to facilitate the problem formulation.
For convenience, we use the term polygon to refer to rectilinear polygons in the rest of this paper.

Let $P$ be an input polygon with $n$ vertices, we define the concave vertices as follows.
\begin{mydefinition}[Concave Vertex]
The concave vertex of a polygon is one at which the internal angle is $270^o$.
\end{mydefinition}

Let $c$ be the number of concave vertices in $P$,
\cite{JOG83_Rourke} gave the relationship between $n$ and $c$: $n = 2 c + 4$.
If the number of concave vertices $c$ is odd, polygon $P$ is called \textit{odd polygon}; otherwise, $P$ is called \textit{even polygon}.

\begin{mydefinition}[Cut]
A cut of a polygon $P$ is a horizontal or vertical line segment at least one of whose endpoints is incident on a concave vertex.
The other endpoint is obtained by extending the line segment inside $P$ until it first encounters the boundary of $P$.
\end{mydefinition}

If both endpoints of a cut are concave vertices in the original polygon, then the cut is called a \textit{chord}.
If a cut has odd number of concave vertices to one side or another, then the cut is called an \textit{odd-cut}.
If an cut is not only odd-cut but also chord, it is called an \textit{odd-chord}.
These concepts are illustrated in Fig. \ref{fig:CutConcept}, where vertices $b, e, h$ are concave vertices,
edges $\bar{bh}, \bar{ej}$ are odd-cuts, and edge $\bar{bh}$ is chord.
Note that $\bar{bh}$ is an odd-chord.

\begin{figure}[bht]
  \centering
  \includegraphics[width=0.22\textwidth]{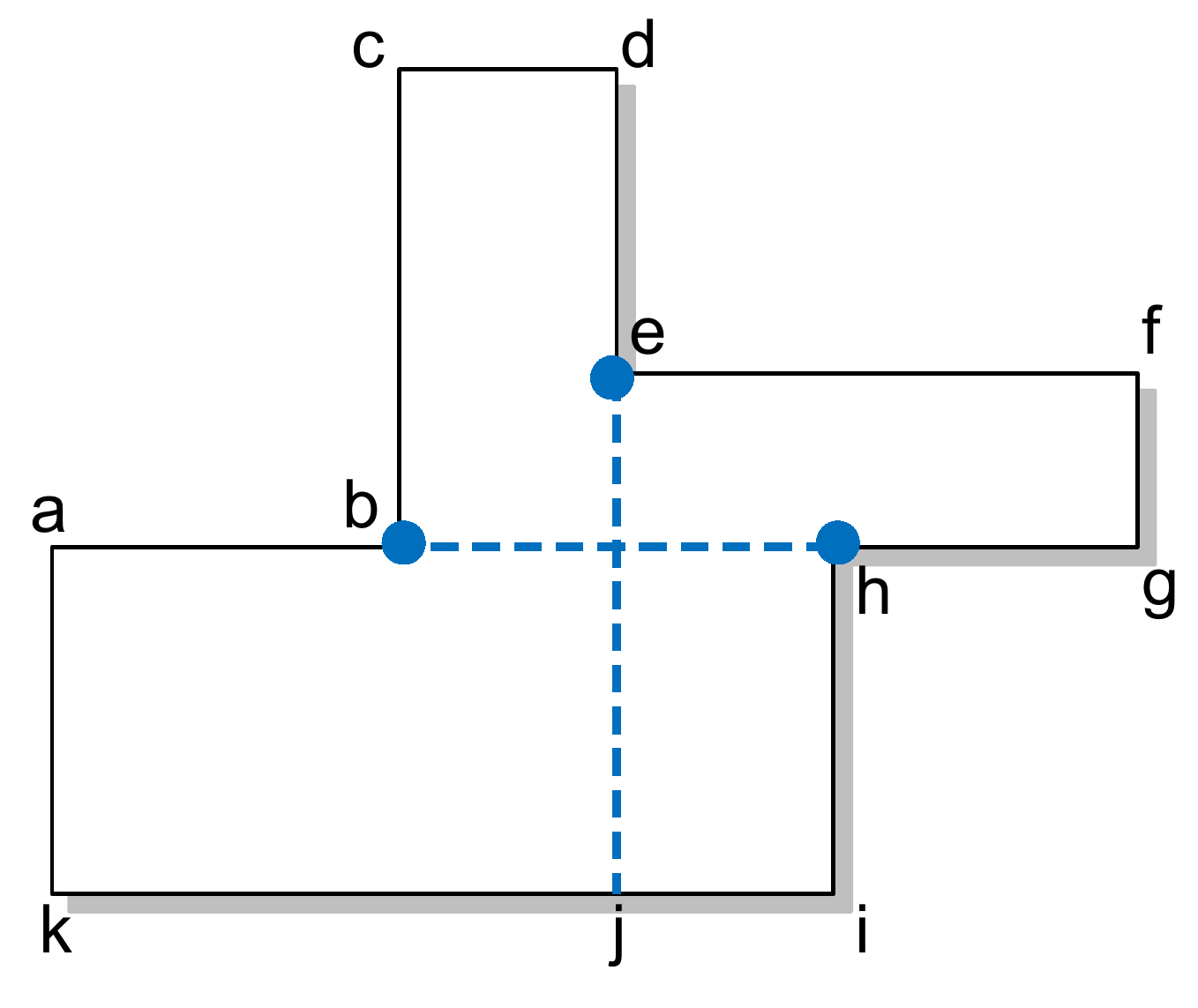}
  \nocaptionrule
  \caption{Concepts of concave vertices and cuts.}
  \label{fig:CutConcept}
  \vspace{-.1in}
\end{figure}

\begin{mydefinition}[L-shape]
An L-shape is a polygon shaped in the form of the letter L.
\end{mydefinition}

An L-shape can be also viewed as a combination of two rectangles with a common coordinate.
There are two easy ways to check whether a polygon is an L-shape.
First, we can check whether the number of vertices equals to $6$, i.e., $n=6$.
Besides, we can check whether there is only one concave vertex, i.e., $c=1$.

\begin{mydefinition}[Sliver Length]
For an L-shape or a rectangle, if the width of its bounding box $B$ is above $\epsilon$, its sliver length is 0.
Otherwise, the sliver length is the length of $B$.
\end{mydefinition}

\begin{problem}[L-shape based Layout Fracturing]
Given an input layout which is specified by polygons, our goal is to fracture it into a set of L-shapes and/or rectangles to minimize the number of shots, and meanwhile minimize the silver length of fractured shots.
\end{problem}

\vspace{-.1in}
\section{Rectangular Merging (RM) Algorithm}
\label{sec:post}

Given the rectangles generated by any rectangular fracturing methodology,
we propose an algorithm, called RM, to merge them into a set of L-shapes.
The main idea is that if two rectangles share a common coordinate, they can be combined into one L-shape.
Although this idea is straightforward, the benefit is obvious that the previous rectangular fracturing algorithms can be re-used.
Besides, the RM algorithm is used as a baseline in comparison with our another algorithm, DLF, which will be described in Section \ref{sec:algo}.

\begin{figure}[bht]
  \centering
  \vspace{-.1in}
  \subfigure[]{
    \includegraphics[width=0.15\textwidth]{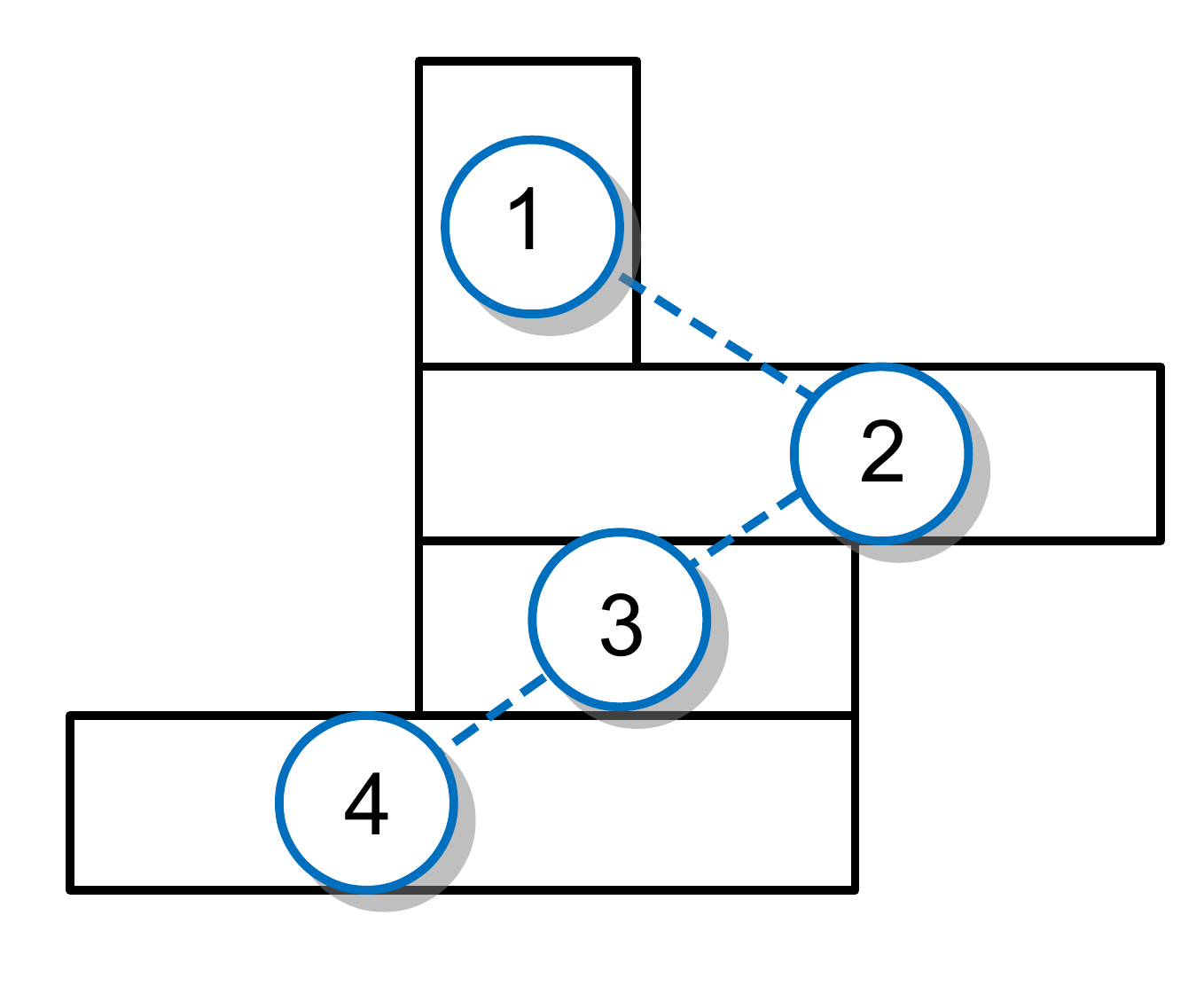}
  }
  \hspace{-.1in}
  \subfigure[]{
    \includegraphics[width=0.15\textwidth]{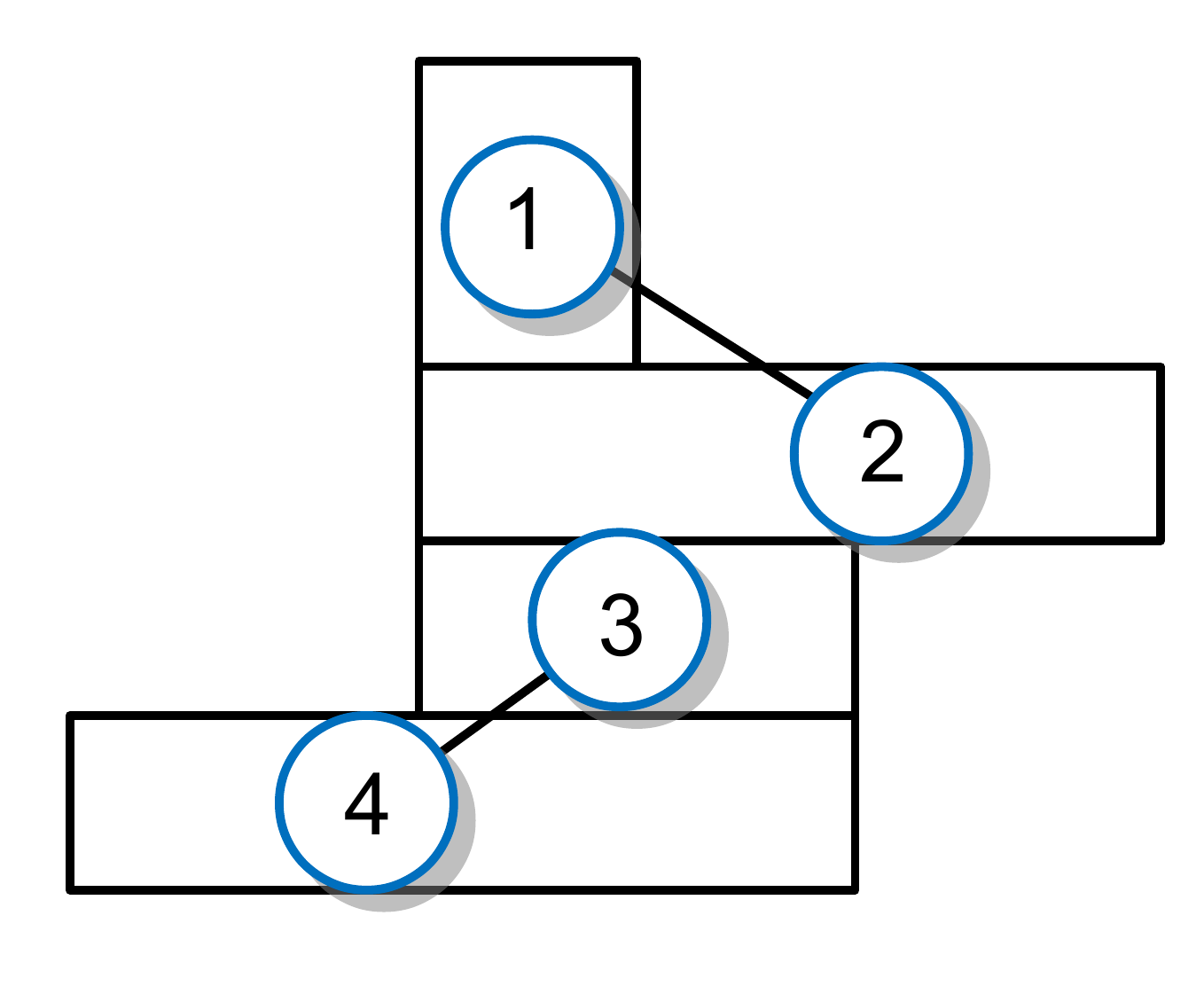}
  }
  \hspace{-.1in}
  \subfigure[]{
    \includegraphics[width=0.15\textwidth]{MergeGraph4}
  }
  \nocaptionrule
  \caption{Example of RM algorithm.~(a) Graph construction.~(b) Maximum matching result.~(c) Corresponding rectangular merging.}
  \label{fig:MergeGraph}
  \vspace{-.1in}
\end{figure}

Given the input rectangles, the RM algorithm can find the optimal L-shape merging solution.
Meanwhile, the shot count and sliver length can be minimized simultaneously.

First we construct a merging graph $G$ to represent the relationships among all the input rectangles.
Each vertex in $G$ represents a rectangle.
There is an edge between two vertices if and only if those two rectangles can be merged into an L-shape.
For example shown in Fig. \ref{fig:MergeGraph}, after rectangular fracturing, four rectangles are generated.
The constructed merging graph $G$ is illustrated in Fig. \ref{fig:MergeGraph}(a),
where the three edges show that there are three ways to generate L-shapes.
L-shape merging can be viewed as edge selection from the merging graph $G$.
Note that one rectangle can only be assigned to one selected edge, that is, no two selected edges share a common end point.
For example, rectangle 2 can only belongs to one L-shape, and thus only one edge can be chosen between edges $\bar{12}$ and $\bar{23}$.

By utilizing the merging graph, the best edge selection can be solved by finding a maximum matching.
Therefore, the rectangular merging can be formulated as a maximum matching problem.
In the case of Fig. \ref{fig:MergeGraph}, the result of the maximum matching is illustrated in Fig. \ref{fig:MergeGraph}(b), and the corresponding L-shape fracturing result is shown in Fig. \ref{fig:MergeGraph}(c).

In order to consider the sliver minimization, we assign weights to the edges to represent whether the merging would remove one sliver.
For example, if there is still one sliver even two rectangles $v_i$ and $v_j$ are merged into one L-shape,
we assign less weight to edge $e_{ij}$. 
Otherwise, larger weight is assigned.
Therefore, the rectangular merging can be formulated as maximum weighted matching.
Even in general graphs, the maximum weighted matching can be solved in $O(nmlogn)$ time \cite{match_1986Galil}, where the $n$ is the number of vertices, and the $m$ is the number of edges in $G$.

\vspace{-.1in}
\section{Direct L-Shape Fracturing (DLF) Algorithm}
\label{sec:algo}

%{{{
%Although the RM algorithm described above is intuitive, it may suffer from several limitations.
Although the RM algorithm described above can provide the optimal merging solution for given rectangles,
it may suffer from several limitations.
%First, solving integer linear programming (ILP) is time-consuming.
First, the polygon is fractured into rectangles first, and followed by a merging stage.
This strategy, however, has some redundant or unnecessary operations.
For the case in Fig. \ref{fig:example}, instead of complex rectangles generation, only one cut is enough for the L-shape fracturing.
Second, the rectangular fracturing may ignore some internal features of L-shape fracturing, which could sacrifices the whole performance.
%although RM algorithm can provide the optimal merging solution for the given rectangles, 
To overcome all these limitations, in this section we propose a novel algorithm, called DLF, to directly fracture polygon into L-shapes.

We observe that the solution space for the L-shape fracturing can be very large.
Given a polygon, there can exist several fracturing solutions with the same shot count.
For example, as shown in Fig. \ref{fig:SolutionSpace}, the input polygon has at least five different fracturing solutions with two shots.

\begin{figure}[bht]
  \centering
  \subfigure[]{\includegraphics[width=0.09\textwidth]{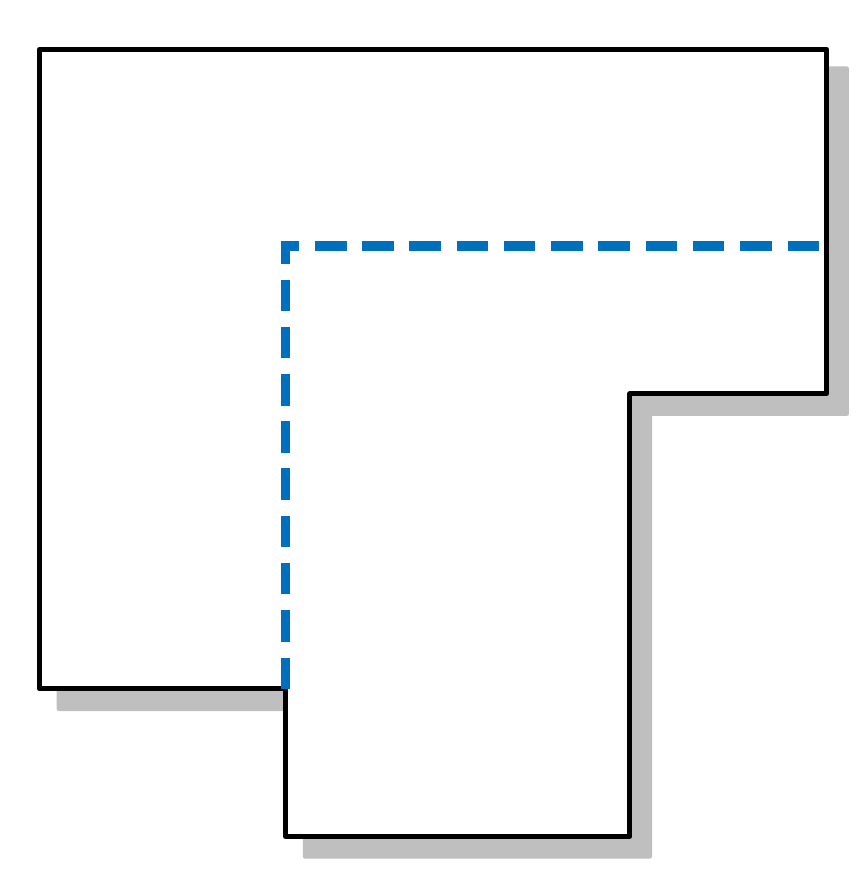}}
  %\hspace{-.1in}
  \subfigure[]{\includegraphics[width=0.09\textwidth]{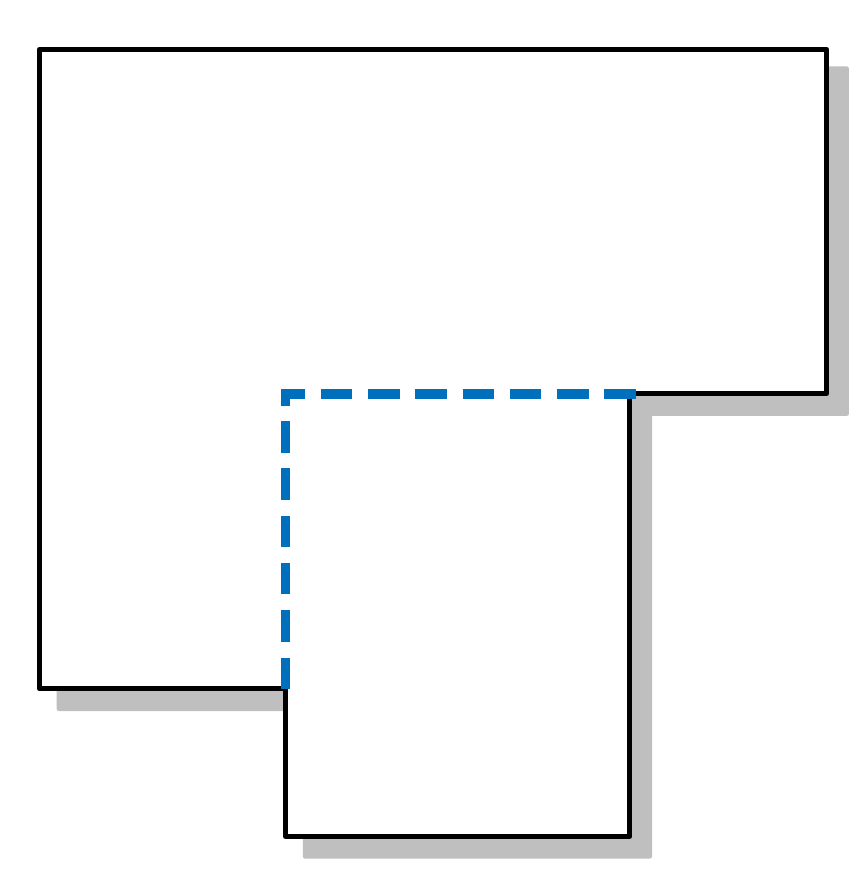}}
  %\hspace{-.1in}
  \subfigure[]{\includegraphics[width=0.09\textwidth]{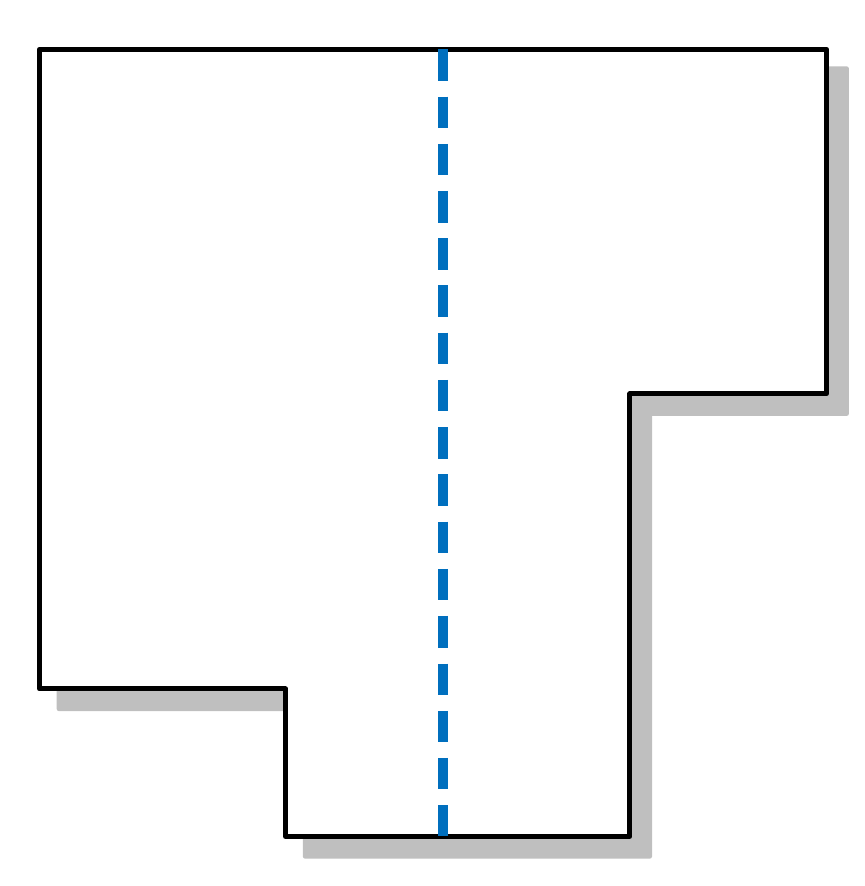}}
  %\hspace{-.1in}
  \subfigure[]{\includegraphics[width=0.09\textwidth]{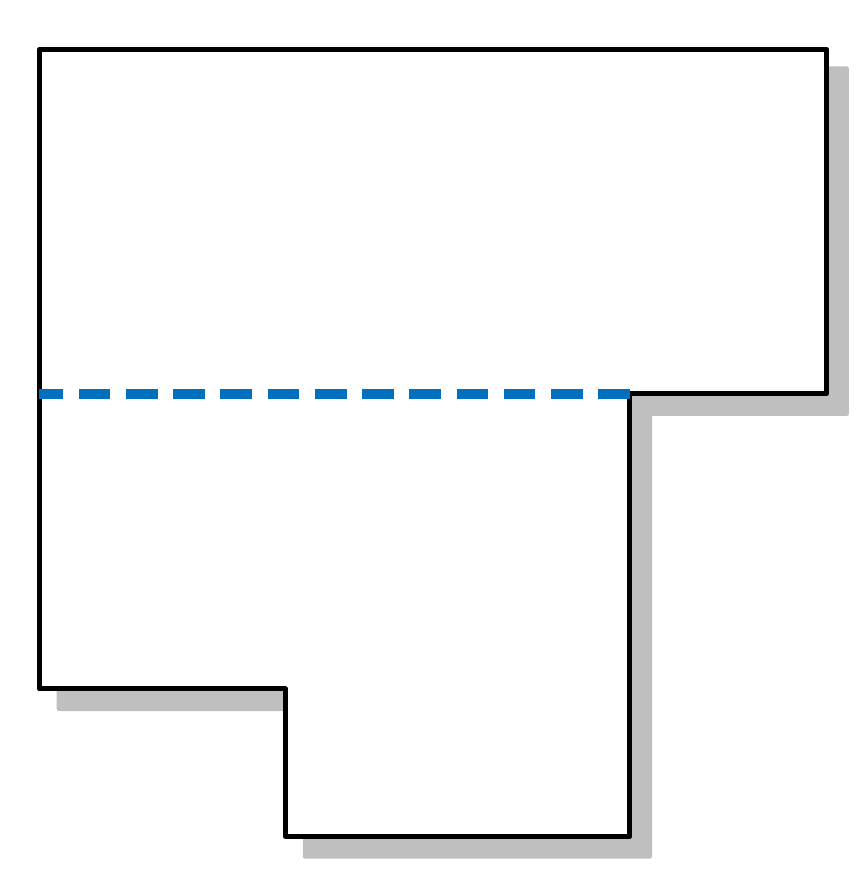}}
  %\hspace{-.1in}
  \subfigure[]{\includegraphics[width=0.09\textwidth]{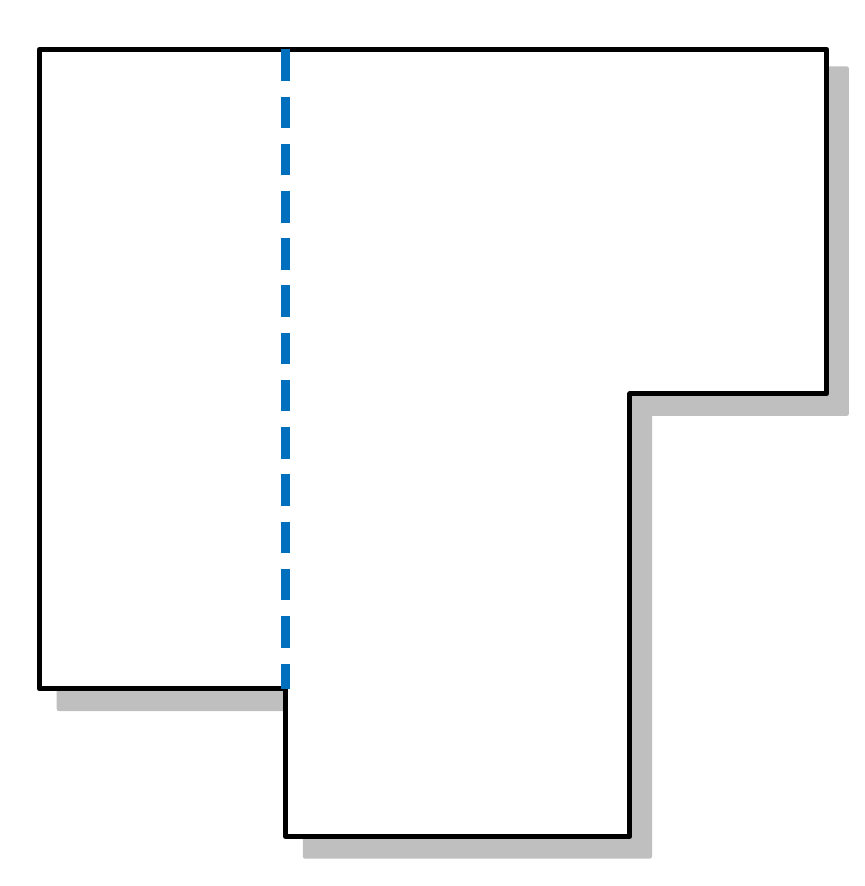}}
  \nocaptionrule
  \caption{Five fracturing solutions with the same shot count.}
  \label{fig:SolutionSpace}
  \vspace{-.1in}
\end{figure}

Note that a cut has the following property: if the polygon is decomposed through a cut,
the concave vertex that is one of the endpoints of the cut is no longer concave in either of the two resulting polygons.
Our L-shape fracturing algorithm, DLF, takes advantage of this property.
Each time a polygon is cut, DLF searches one appropriate odd-cut to decompose the polygon.
%The main idea is that each time we cut the polygon; we try to decompose it through one of its odd-cuts.
It was shown in \cite{JOG83_Rourke} that odd-cut always exists
and $\lfloor c/2 \rfloor + 1$ ``guards'' are necessary and sufficient to cover all the interiors of a polygon with $c$ concave vertices.
Therefore, we can obtain the following lemma.
%Since each guard can cover at most one L-shape, we can achieve the following lemma.

\begin{mylemma}
\label{lem:concave_vertice}
A polygon with $c$ concave vertices can be decomposed into L-shapes with upper bound number $N_{up} = \lfloor c/2 \rfloor +1$.
\end{mylemma}

\begin{figure}[bht]
  \centering
  \vspace{-.1in}
  \includegraphics[width=0.42\textwidth]{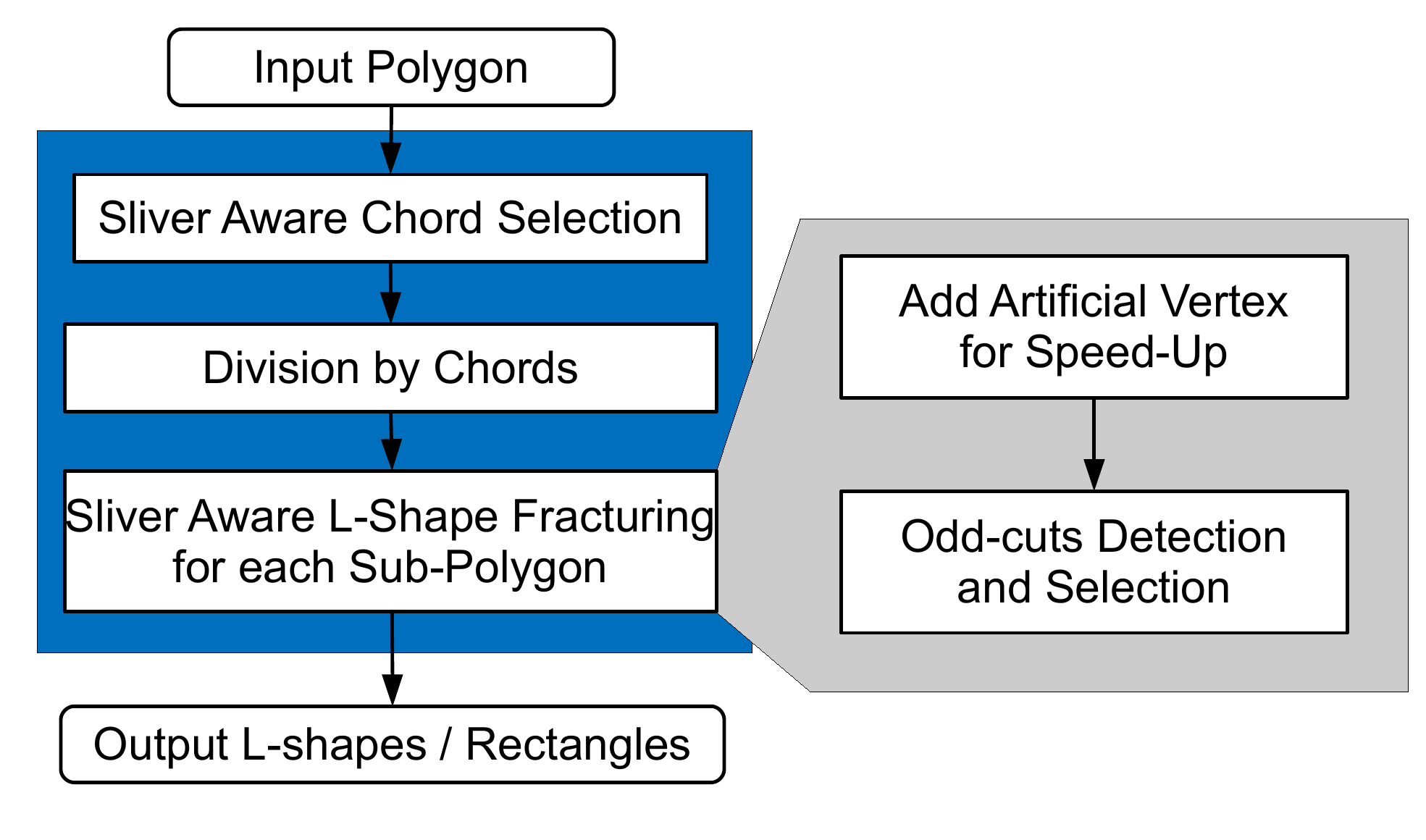}
  \nocaptionrule
  \caption{Overall flow of DLF algorithm.}
  \label{fig:flow}
  \vspace{-.1in}
\end{figure}

Fig. \ref{fig:flow} shows the overall flow of our DLF algorithm. We will effectively use chords and cuts to reduce the problem size while containing or even reducing  the L-shape fracturing number upper bound.  The first step is to detect all chords (i.e., horizontal or vertical cuts made by concave points), in particular odd-chords as they may reduce the L-shape upper bound.  We will then perform sliver-aware chord selection to decompose the original polygon $P$ into a set of sub-polygons. Then for each sub-polygon, we will perform sliver aware L-shape fracturing, where odd-cuts are detected and selected to iteratively cut the polygon into a set of L-shapes. The reason we differ chord and cut during polygon fracturing is that chord is a special cut with both end points being concave points in the original polygon. That way, we can design more efficient algorithm for odd cut/chord detection.

%}}}

% ========================================================
%                   Chord Selection
% ========================================================
\vspace{-.1in}
\subsection{Sliver Aware Chord Selection}
%{{{

The first step of DLF algorithm is sliver aware chord selection.
Cutting through chords decomposes the whole polygon $P$ into a set of sub-polygons. By this way the problem size is reduced.
We can further prove that cutting through a chord does not increase the L-shape upper bound $N_{up}$.

\begin{mylemma}
\label{lem:chord}
Decompose a polygon through a chord does not increase the L-shape upper bound number $N_{up}$.
\end{mylemma}
\begin{proof}
Cut the polygon along this chord, and let $c_1$ and $c_2$ be the number of concave vertices in the two pieces produced.
Since $c = c_1 + c_2 +2$, then using Lemma~\ref{lem:concave_vertice} we have
\begin{eqnarray}
  \lfloor c_1/2 \rfloor + 1 + \lfloor c_2/2 \rfloor + 1 & \le & \lfloor (c_1 + c_2)/2 \rfloor + 2  \notag\\
                                                        &  =  & \lfloor (c-2)/2 \rfloor + 2 = \lfloor c/2 \rfloor + 1            \notag
\end{eqnarray}
\end{proof}

%Based on Lemma \ref{lem:chord}, we shall pick as many chords as possible. It shall be noted that one chord may intersect with another thus they are not compatible. 
Chord selection has been proposed in rectangular fracturing \cite{EBL_SPIE06_Kahng}\cite{EBL_SPIE2011_Ma}, but for L-shape fracturing, odd-chord shall be selected as they can even reduce the number of L-shapes. 

\begin{mylemma}
\label{lem:odd-chord}
Decomposing a even polygon along an odd-chord can reduce the L-shape upper bound number $N_{up}$ by $1$.
\end{mylemma}

The proof is similar to that for Lemma \ref{lem:chord}.
The only difference is that since $c$ is even and $c_1, c_2$ are odd,
$\lfloor c_1/2 \rfloor + 1 + \lfloor c_2/2 \rfloor + 1 < \lfloor c/2 \rfloor + 1$.
%The reason is that since $c$ is even, $\lfloor c/2 \rfloor$ 
Note that for an odd polygon, all chords are odd. For a even polygon, Lemma \ref{lem:odd-chord} provides a guideline to select chords.
An example is illustrated in Fig. \ref{fig:LemmaThree}, which contains two chords $\bar{bh}$ and $\bar{hk}$.
Since the number of concave vertices to both side of chord $\bar{bh}$ are odd ($1$), $\bar{bh}$ is an odd-chord.
Cut along $\bar{bh}$, as shown in Fig. \ref{fig:LemmaThree}(a), can achieve two L-shots.
However, cut along another chord $\bar{hk}$, which is not an odd-chord, would need three shots.
Note that in an odd polygon, although all chords are odd, cutting along them may not reduce $N_{up}$, but it will not increase $N_{up}$ either.

\begin{figure}[bht]
  \centering
  \subfigure[]{\includegraphics[width=0.22\textwidth]{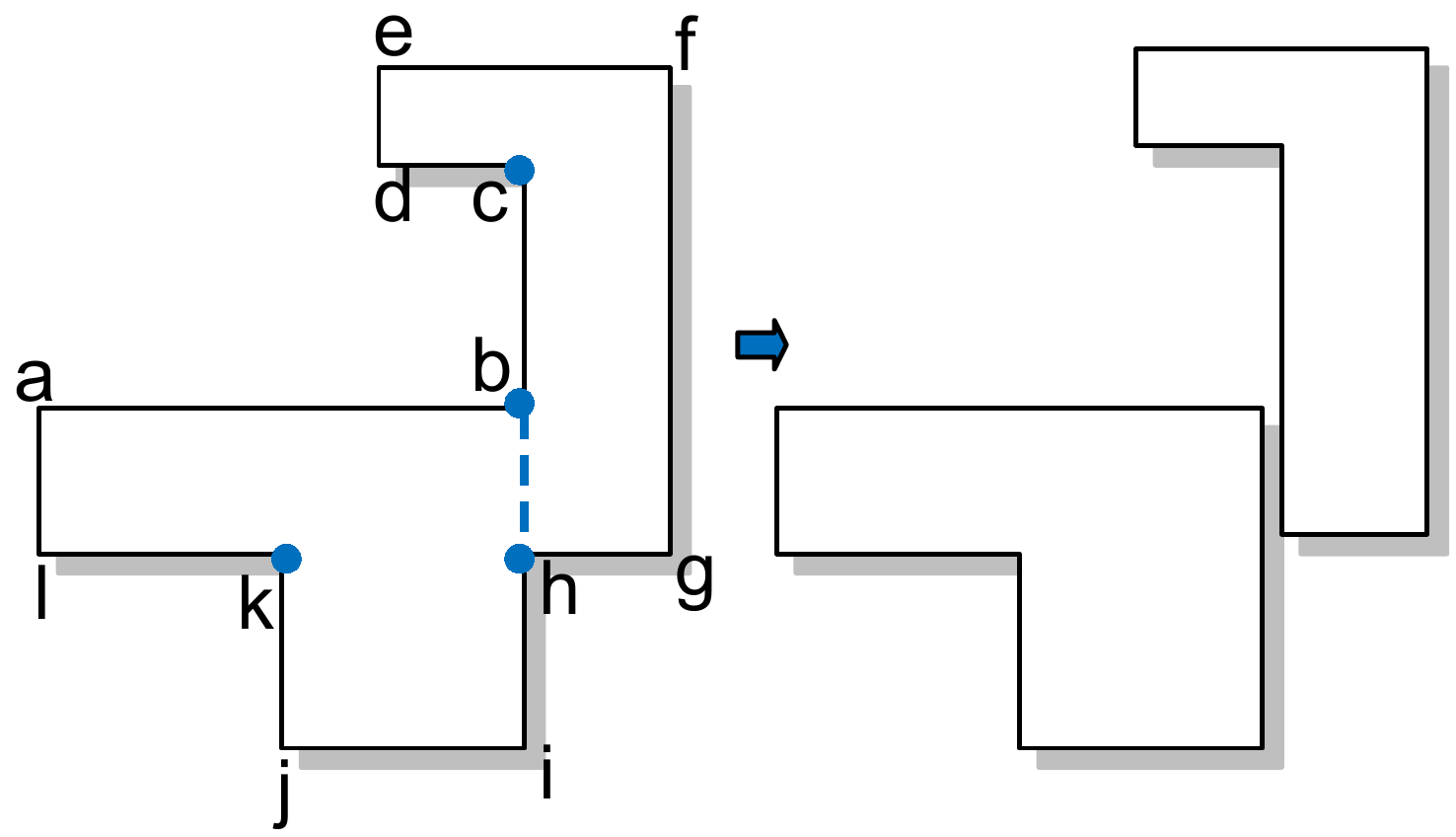}}
  %\hspace{-.1in}
  \subfigure[]{\includegraphics[width=0.22\textwidth]{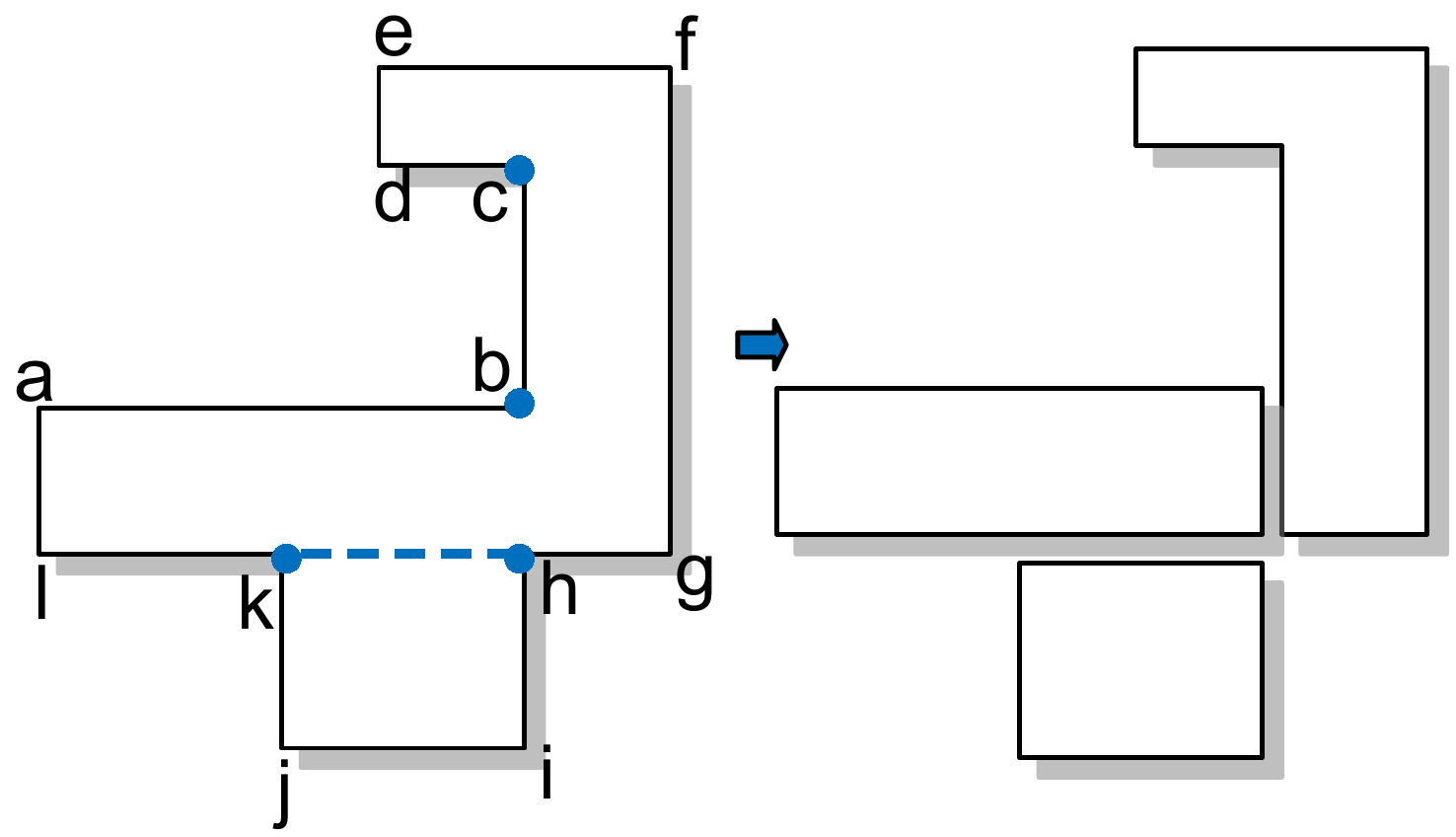}}
  \nocaptionrule
  \caption{Examples to illustrate Lemma \ref{lem:odd-chord}.~(a) Cut along odd-chord $\bar{bh}$ results in two L-shape shots.~(b) Cut along chord $\bar{hk}$ would cause one more shot.}
  \label{fig:LemmaThree}
  \vspace{-.1in}
\end{figure}

%Here we show how to find all odd-chords effectively.
For any even polygon $P$, we propose the odd-chord search procedure as follows.
Each vertex $v_i$ is assigned with one Boolean parity $p_i$.
Starting from an arbitrary vertex with any parity assignment, we proceed clockwise around the polygon.
If the next vertex $v_j$ is concave, then $p_j = \neg p_i$, where $p_i$ is the parity of current vertex $v_i$.
Otherwise $p_j$ is assigned to $ p_i$.
This parity assignment can be completed during one clockwise traverse in $O(n)$ time.
An example of parity assignment starting from a(0) is shown in Fig. \ref{fig:OddChordDetection}, where each vertex is associated with one parity.

\begin{figure}[bht]
  \centering
  \vspace{-.1in}
  \includegraphics[width=0.2\textwidth]{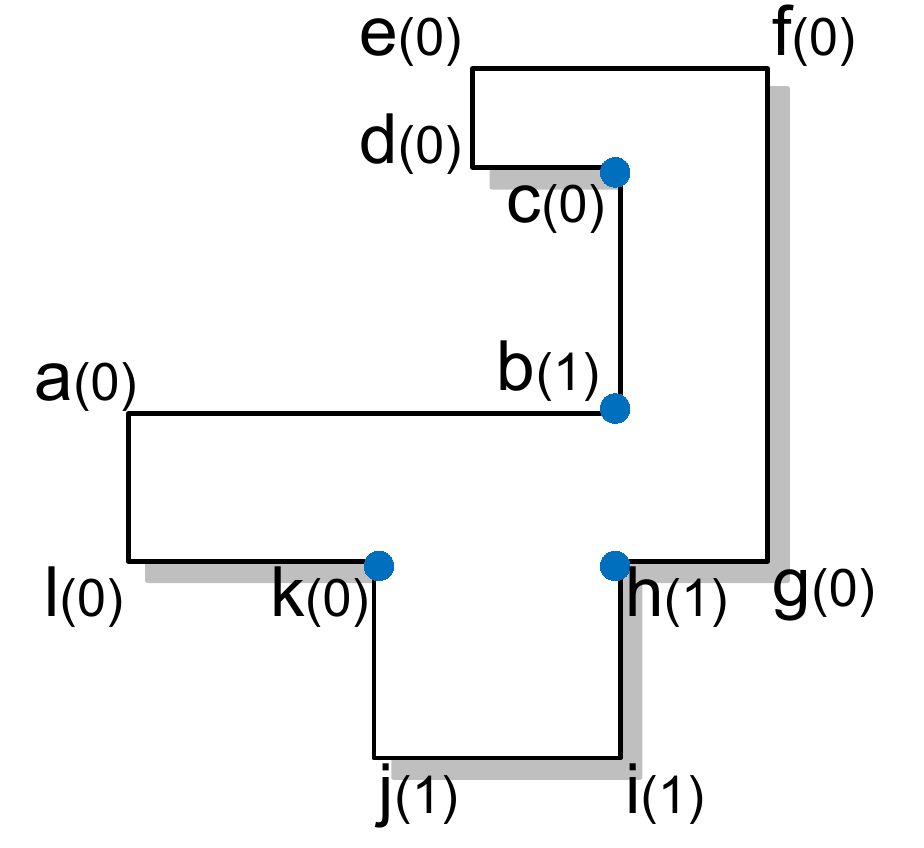}
  \nocaptionrule
  \caption{To detect odd-chords in even polygon, each vertex is associated with one Boolean parity.}
  \label{fig:OddChordDetection}
  \vspace{-.1in}
\end{figure}

\begin{mytheorem}
\label{theo:odd-chord}
In an even polygon, a chord $\bar{ab}$ is odd iff $p_a = p_b$.
\end{mytheorem}

Given the parity values, the odd-chord detection can be performed using Theorem \ref{theo:odd-chord}.
For each concave vertex $v_i$, a plane sweep is applied to search any chord containing $v_i$.
The plane sweep for each vertex can be finished in $O(logn)$, and the number of vertices is $O(n)$.
Therefore, in an even polygon odd-chords detection can be completed in $O(nlogn)$ time.

After all chords are detected, chord selection is applied to choose as many chords as possible to divide the input polygon into a set of independent sub-polygons.
Note that if a chord is selected to cut the polygon, it releases the concavity of its two endpoints.
Therefore, if two chords intersect with each other, at most one of them could be selected.
For example, in Fig. \ref{fig:OddChordDetection}, chords $\bar{bh}$ and $\bar{hk}$ cannot be selected simultaneously.
The relationship among the chords can be represented as a bipartite graph \cite{EBL_SPIE06_Kahng},
and the vertices in left and right columns indicate the horizontal and vertical chords, respectively.
Therefore, finding the most chords compatible with each other corresponds to finding the maximum independent set in the bipartite graph,
which can be reduced to maximum matching problem, and therefore, can be done in polynomial time.
It shall be noted that if the input polygon is a even polygon, because of Theorem \ref{theo:odd-chord}, we prefer to choose odd-chords.
Therefore, the bipartite graph is modified by assigning weights to the edges.
In addition, sliver minimization is integrated into the chord selection.
When an odd-chord candidate is detected, we calculate the distance between it and the boundary of the polygon.
If the distance is less than $\epsilon$, cutting this odd-chord would cause sliver, then we will discard this candidate.

%}}}

% ========================================================
%         Hierarchical Odd-Cut Selection
% ========================================================
\vspace{-.1in}
\subsection{Sliver Aware L-Shape Fracturing}
%{{{

After chord selection, the input polygon $P$ is decomposed into $m$ sub-polygons (denoted as $P_1, P_2, ..., P_m$).
For each polygon $P_i$, we will recursively fracture it until reaching final L-shapes and/or rectangles.
Our fracturing algorithm is based on odd-cut selection.
The main idea is that each time we pick up one odd-cut, and decompose the polygon into two pieces through this odd-cut.
Iteratively we fracture the polygon into several L-shapes.
Note that our fracturing algorithm considers the sliver minimization, i.e., we try to minimize the sliver length during fracturing.

The first question is how to detect all the odd-cuts efficiently.
%In other words, given a cut in polygon $P_i$, we need an effective method to check whether it is an odd-cut of $P_i$.
%Note that if $P_i$ is even polygon, it is easy to see that for each cut, if it is not a chord, it is an odd-cut, thus the detection is completed. 
%However, if $P_i$ is odd polygon, the odd-cuts detection is not so straightforward.
%
Our method is similar to that for odd-chord detection.
Each vertex $v_i$ is assigned an order number $o_i$, and a Boolean parity $p_i$.
Start at an arbitrary vertex, each vertex $v_i$ is assigned an order $o_i$.
We initialize the Boolean parity $p$ to zero, and proceed clockwise around the polygon.
If the next vertex $v_i$ is normal, label its $p_i$ as $p$;
if $v_i$ is concave, assign $p$ to $\neg p$, and label its $p_i$ with the new $p$ value.
For each concave vertex $v_a$, we search cuts from two directions (horizontal and vertical) from it.
Here we denote $(a,\bar{bc})$ as the cut with one endpoint at vertex $v_a$ and the other endpoint at edge $\bar{bc}$.
For each cut $(a, \bar{bc})$ detected, whether it is an odd-cut can be checked in constant time using the following Theorem \ref{theo:odd-cut}.

\begin{mytheorem}
\label{theo:odd-cut}
In an odd polygon, a cut $(a, \bar{bc})$ is an odd-cut if and only if the following condition is satisfied:
\begin{displaymath}
	\left\{
	\begin{array}{cc}
		p_a = p_b, 			  & \textrm{if}\ \  o_a > o_b\\
		p_a \neq p_b, 		  & \textrm{if}\ \  o_a < o_b
	\end{array}
	\right.
\end{displaymath}
\end{mytheorem}

Due to space limit, the detailed proof is omitted.
An example of odd-cut detection is shown in Fig. \ref{fig:OddCutDetection}.
There are three concave vertices, $v_b, v_f$ and $v_i$ in the odd polygon.
Start from each concave vertex, we have searched all six cuts.
Applying Theorem \ref{theo:odd-cut}, we find out two odd-cuts $(b, \bar{fg})$ and $(i, \bar{cd})$.

\begin{figure}[bht]
  \centering
  \vspace{-.1in}
  \includegraphics[width=0.24\textwidth]{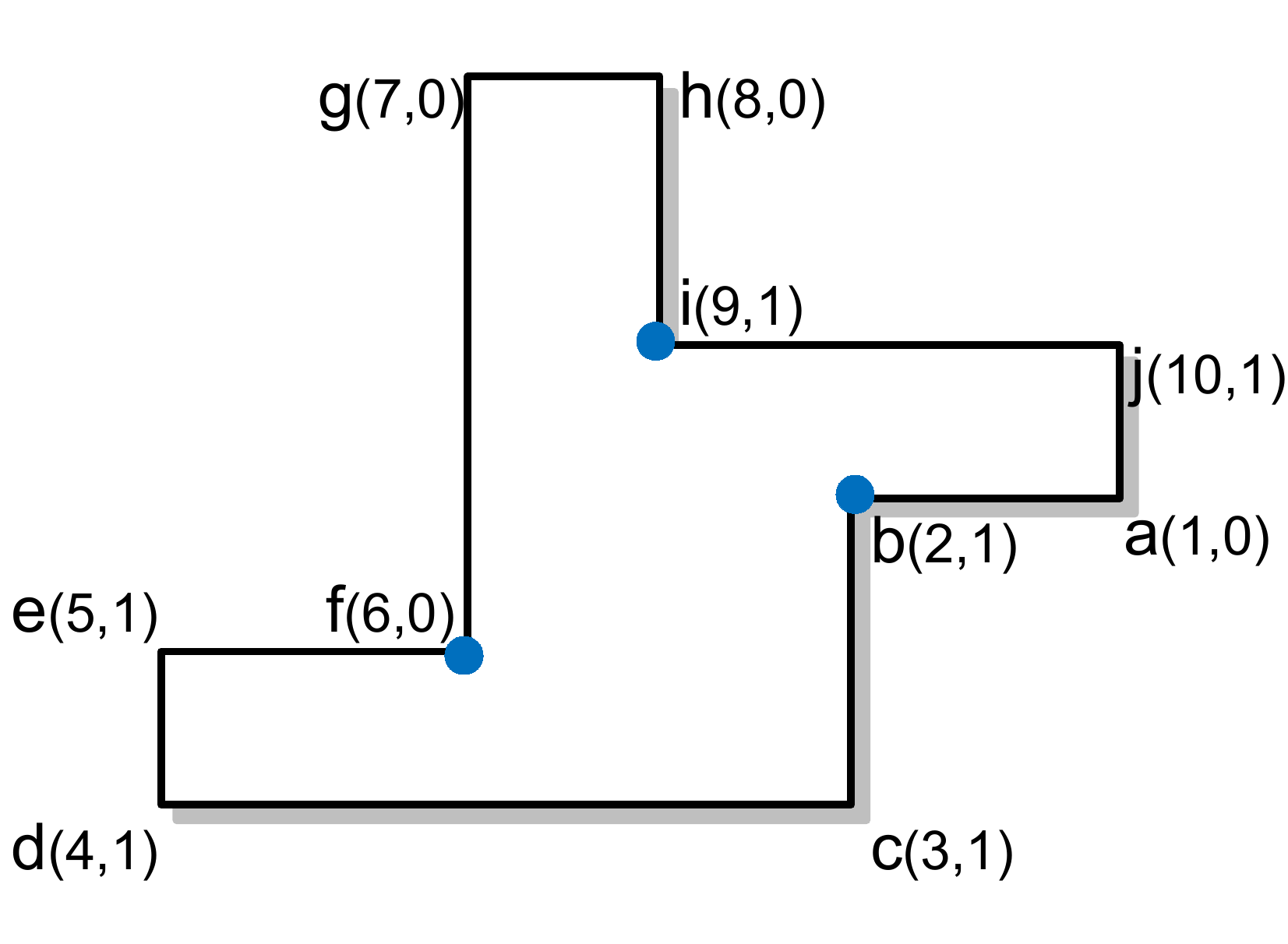}
  \nocaptionrule
  \caption{Odd-cut detection using order number and parity.}
  \label{fig:OddCutDetection}
  \vspace{-.1in}
\end{figure}

\begin{algorithm}[htb]
\caption{LShapeFracturing($P$)}
\label{alg:fracturing}
\begin{algorithmic}[1]
  \REQUIRE Polygon $P$.
  \IF{ $P$ is L-shape or rectangle}
    \STATE Output $P$ as one of results;
    \RETURN
  \ENDIF
  \STATE Find all odd-cuts;
  \STATE Choose cut $cc$ considering the sliver minimization;
  \IF{Cannot find legal odd-cut}
    \STATE Generate an auxiliary cut $cc$;
  \ENDIF
  \STATE Cut $P$ through $cc$ into two polygons $P1$ and $P2$;
  \STATE Update one vertex and four edges;
  \STATE LShapeFracturing($P1$);
  \STATE LShapeFracturing($P2$);
\end{algorithmic}
\end{algorithm}

The details of our L-shape fracturing are described in Algorithm \ref{alg:fracturing}.
Given the input polygon $P$, if it is already an L-shape or rectangle, then the fracturing is completed.
Otherwise, we find all odd-cuts as described above (line $5$).
%If all odd-cuts would cause sliver length penalty, we generate one auxiliary cut (lines $6 - 9$).
From all the odd-cuts detected, we choose one $cc$, and cut the $P$ into two pieces $P_1$ and $P_2$ (lines $10 - 11$).
Then we recursively apply L-shape fracturing to $P_1$ and $P2$ (lines $12 - 13$).

Note that during the polygon decomposition, we do not need to re-calculate the order number and parity of each vertex.
Instead, when a polygon is divided into two parts, we only update one vertex and four edges, while all other information can be maintained.
If polygon $P$ is cut through odd-cut $(a, \bar{bc})$, a new vertex, namely $d$, is generated.
For the new vertex $d$, its order number $o_d = o_b$ and its parity $p_d = p_b$.
Edge $\bar{bc}$ is replaced by edges $\bar{bd}$ and $\bar{dc}$.
Besides, two edges $\bar{ad}$ and $\bar{da}$ are inserted.
The update method is simple and easy to implement.
An example of such update is shown in Fig. \ref{fig:update}.

\begin{figure}[bht]
  \centering
  \vspace{-.1in}
  \includegraphics[width=0.44\textwidth]{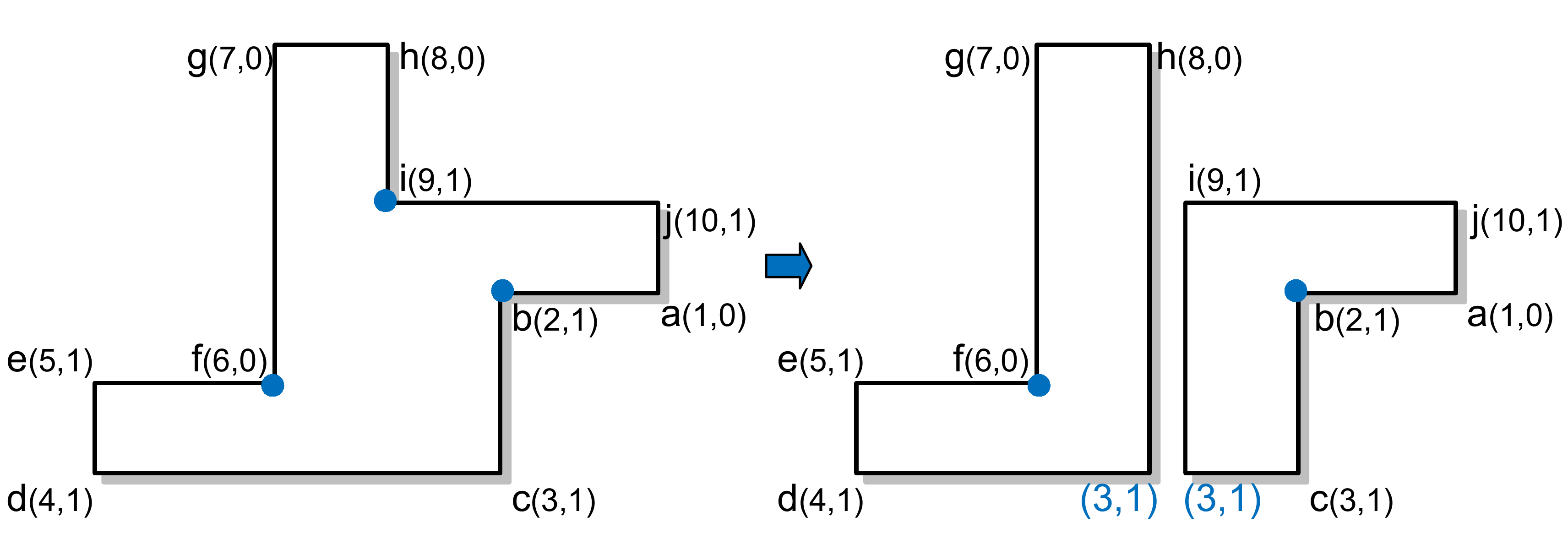}
  \nocaptionrule
  \caption{Only one vertex and four edges need to be updated during polygon decomposition.}
  \label{fig:update}
  \vspace{-.1in}
\end{figure}

Sliver minimization is integrated into our L-shape fracturing algorithm.
In Algorithm \ref{alg:fracturing}, when picking up one cut from all odd-cuts, we try to avoid any sliver.
For example, as illustrated in Fig. \ref{fig:SliverAware}(a), there are three odd-cuts, but all of them would cause sliver.
Instead of selecting any of them, we generate an auxiliary cut in the middle (see Fig. \ref{fig:SliverAware}(b)).
Because of the auxiliary cut, the polygon can be fractured without introducing any sliver.
In addition, if there are several odd-cuts not causing sliver, we pick the cut using the following rules:
(1) We prefer the cut which partitions the polygon into two balanced sub-polygons;
(2) If the polygon is more horizontal than vertical, we prefer a vertical cut, and vice verse.

\begin{figure}[bht]
  \centering
  \vspace{-.1in}
  \subfigure[]{\includegraphics[width=0.21\textwidth]{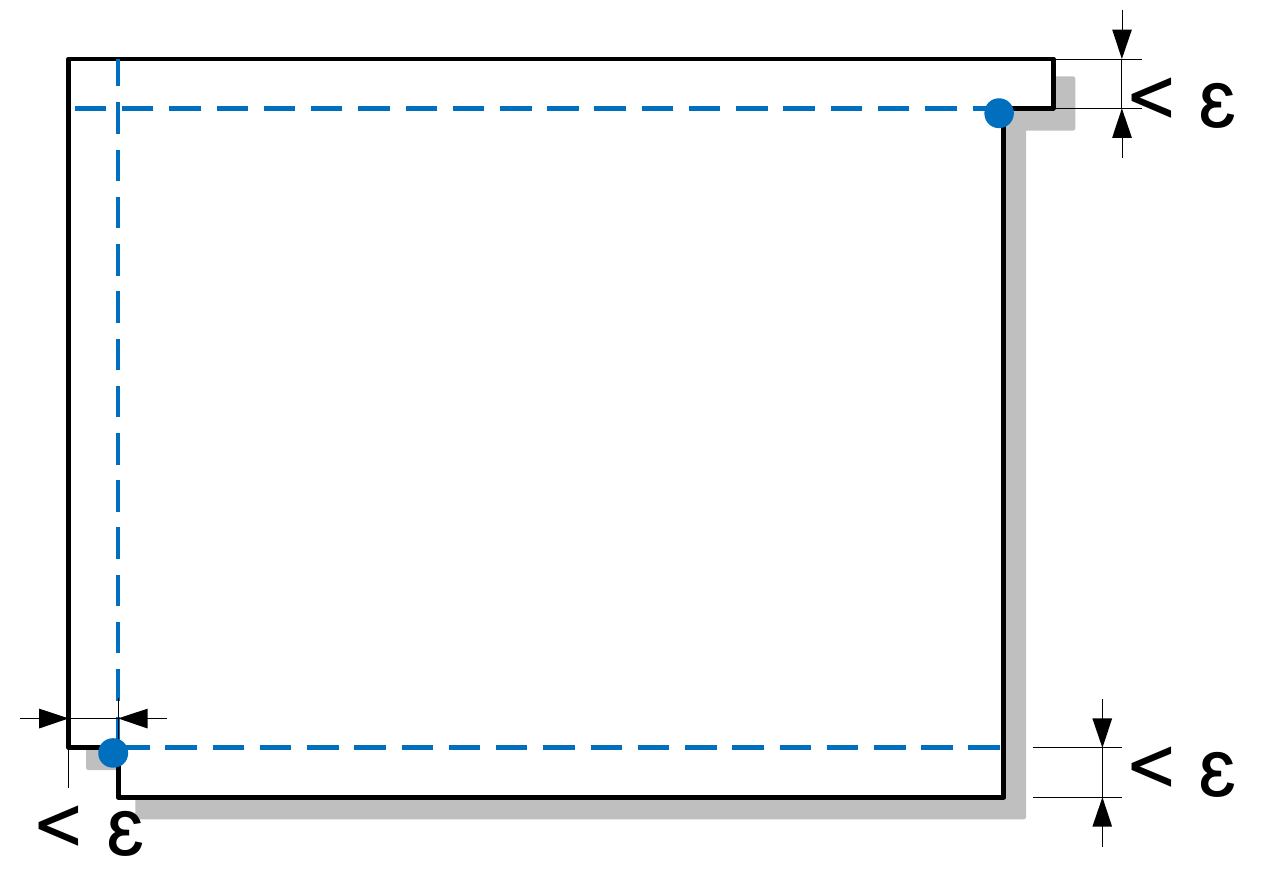}}
  \subfigure[]{\includegraphics[width=0.21\textwidth]{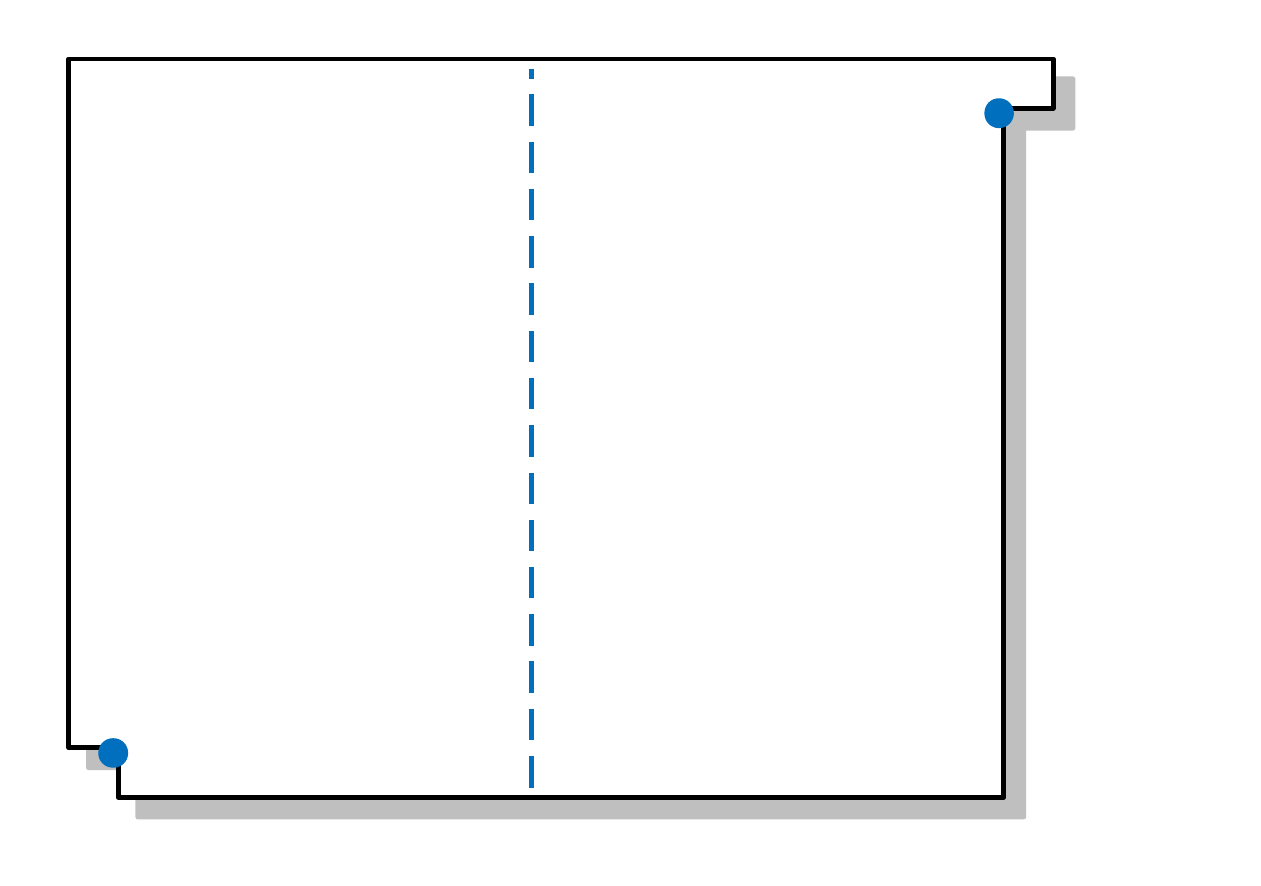}}
  \nocaptionrule
  \caption{Auxiliary cut generation.~(a) Here every odd-cut would cause sliver.~(b) Decompose through on auxiliary cut can avoid sliver.}
  \label{fig:SliverAware}
  \vspace{-.1in}
\end{figure}

Given a polygon with $n$ vertices, finding all concave vertices need $O(n)$ time.
For each concave vertex $v_i$, searching cut starting from it needs $O(logn)$ time.
Using Theorem \ref{theo:odd-cut}, checking whether the cut is odd-cut needs $O(1)$,
thus finding all odd-cuts needs $O(nlogn)$ time.
Note that given a polygon with $c$ concave vertices, if no auxiliary cut is generated,
the L-shape fracturing can be completed through $\lfloor c/2 \rfloor$ odd-cuts.
When auxiliary cuts are applied, there are at most $c-1$ cuts to fracture the input polygon.
Therefore, we can achieve the following theorem.

\begin{mytheorem}
The sliver aware L-shape generation can find a set of L-shapes in $O(n^2logn)$ time.
\end{mytheorem}

It shall be noted that if our objective is only to minimize the shot number,
no auxiliary cut would be introduced, thus at most $\lfloor c/2 \rfloor + 1$ L-shapes are generated.
In other words, the shot number would be less or equals to the theoretical upper bound $N_{up}$.

%}}}

% ========================================================
%               Acceleration Technique
% ========================================================
\vspace{-.1in}
\subsection{Speedup Technique}
%{{{

We observe that in practice during the execution of Algorithm \ref{alg:fracturing}, many odd-cuts do not intersect.
In other words, many odd-cuts are compatible, and could be used to decompose the polygon at the same time.
Instead of only picking one odd-cut at one time, we can achieve further speed-up by selecting multiple odd-cuts simultaneously.

If the polygon is an odd polygon, this speed-up is easily implemented.
In the odd polygon, there is only one type of odd-cut: a cut that has an odd number of concave vertices to each side.
Partitioning the polygon along such odd-cut can leave all other odd-cuts remaining to be odd-cuts.
For example, Fig. \ref{fig:SpeedupOdd}(a) shows an odd polygon, where all three odd-cuts are compatible, and can be picked up simultaneously.
Through fracturing the polygon along the three odd-cuts, the L-shape fracturing problem is resolved directly.

\begin{figure}[bht]
  \centering
  \vspace{-.1in}
  \includegraphics[width=0.25\textwidth]{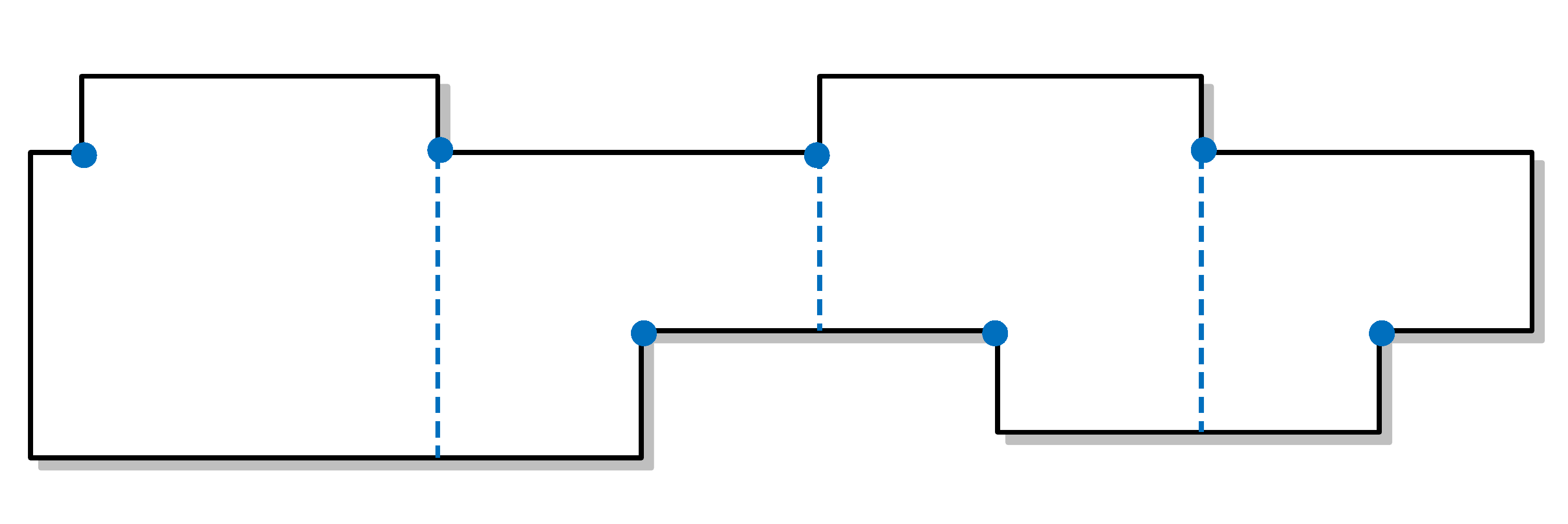}
  \nocaptionrule
  \caption{Speed-up for odd polygon, where all three odd-cuts are compatible.}
  \label{fig:SpeedupOdd}
  \vspace{-.1in}
\end{figure}

However, this speed-up technique cannot be directly applied to an even polygon, since it may cause more shot number.
The reason is that when an even polygon is cut into two pieces, some odd-cuts may no longer be odd-cuts in the sub-polygons.
For example, as shown in Fig. \ref{fig:SpeedupEven}(a), in this even-polygon all six cuts are odd-cuts and compatible.
However, if we use all these compatible cuts for fracturing, we would end up with seven rectangular shots, which is obviously sub-optimal.
To overcome this issue, for each even-polygon we introduce one artificial concave vertex.
Through this artificial concave vertex, the polygon is translated into an odd polygon.
Because of Lemma \ref{lem:artificial}, this translation does not increase the total shot number.
As shown in Fig. \ref{fig:SpeedupEven}(b), in the modified odd polygon,
all compatible odd-cuts can be used for fracturing without causing more shot number.

\begin{mylemma}
\label{lem:artificial}
Introducing one artificial concave vertex to an even polygon does not increase the L-shape upper bound $N_{up}$.
\end{mylemma}

\begin{figure}[bht]
  \centering
  \vspace{-.1in}
  \hspace{-.1in}
  \subfigure[]{\includegraphics[width=0.23\textwidth]{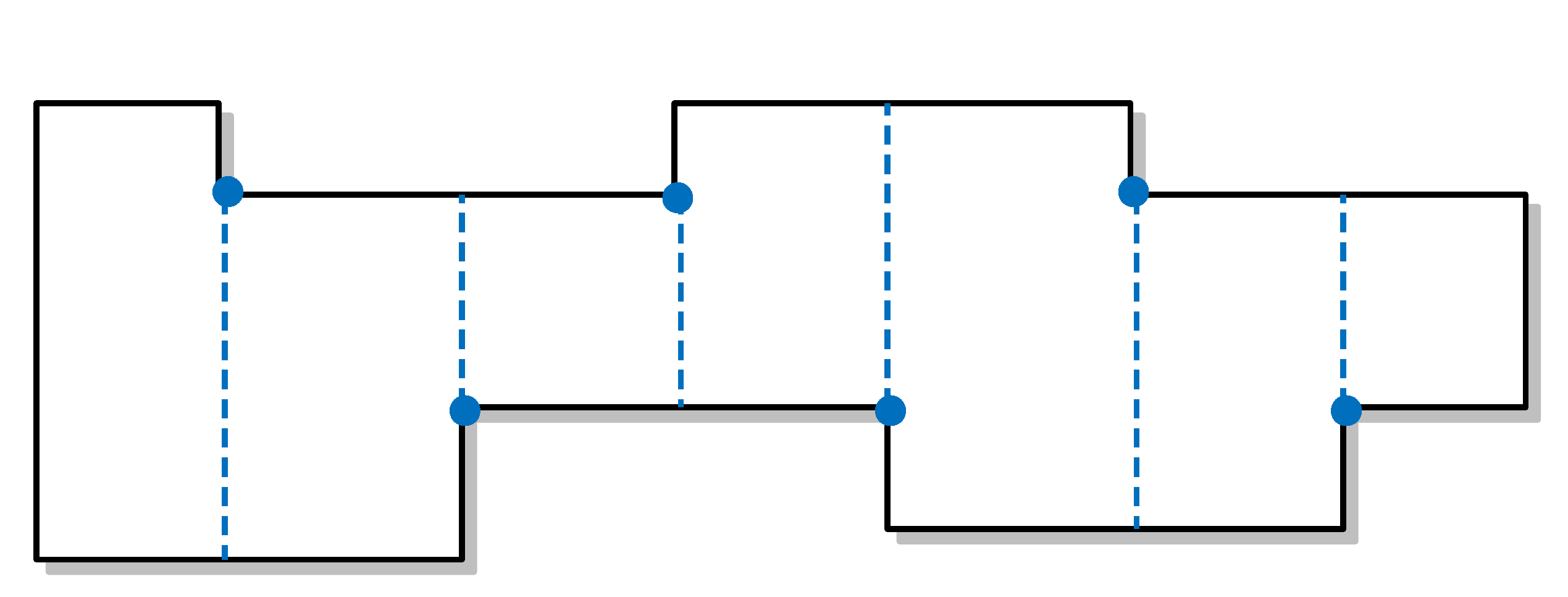}}
  \hspace{-.1in}
  \subfigure[]{\includegraphics[width=0.23\textwidth]{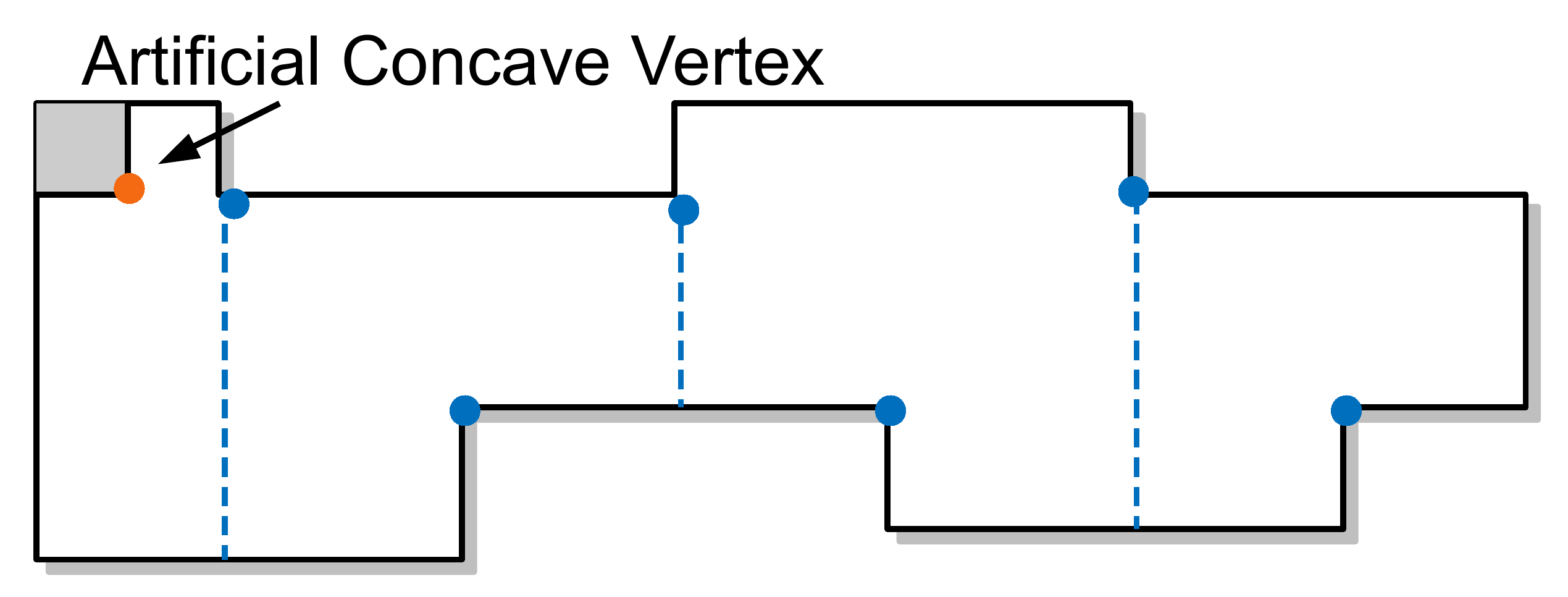}}
  \nocaptionrule
  \caption{Speed-up for even polygon.~(a) all cuts are odd-cuts.~(b) Introducing one artificial concave vertex, translate the even polygon into an odd polygon.}
  \label{fig:SpeedupEven}
  \vspace{-.1in}
\end{figure}

%\begin{mytheorem}
%The L-shape fracturing can be finished in $O(nlogn)$ time using speed-up technique.
%\end{mytheorem}

Through employing this speed-up technique, for most cases, the odd-cut detection can be applied only once,
therefore the DLF algorithm could be completed in $O(nlogn)$ time in practice.

%}}}

\vspace{-.1in}
\section{Experimental Results}
\label{sec:result}

We implemented our two L-shape fracturing algorithms, RM and DLF, in C++.
Since RM needs a rectangular fracturing method to generate initial rectangles,
we implemented a state-of-the-art algorithm proposed in \cite{EBL_SPIE2011_Ma}.
Based on the generated rectangles, RM algorithm is applied to merge them into a set of L-shapes.
%Gurobi \cite{Gurobi} was employed as our ILP solver.
LEDA package \cite{book99_LEDA} is adopted for the maximum weighted matching algorithm.

The experiments are performed on an Intel Xeon 3.0GHz Linux machine with 32G RAM.
ISCAS 85\&89 benchmarks are scaled down to 28nm logic node, followed by accurate lithographic simulations performed to the Metal $1$ layers.
All involved lithography simulations in the \textit{Calibration Phase} are applied under industry-strength RET (OPC).
For all the resulting post-OPC layers, OpenAccess 2.2 \cite{OpenAccess} is adopted for interfacing.

% RESULT TABLE
%{{{
\begin{table*}[bth]
\footnotesize
\centering
\nocaptionrule
\caption{Runtime and Performance Comparisons}
\label{tab:compare}
\begin{tabular}{|c|c|c|c|c||c|c|c||c|c|c|}
  \hline \hline
  \multirow{2}{*}{Circuts}  &\multirow{2}{*}{poly\#}
  &\multicolumn{3}{c||}{\cite{EBL_SPIE2011_Ma}} &\multicolumn{3}{c||}{\cite{EBL_SPIE2011_Ma}+RM} & \multicolumn{3}{c|}{DLF}\\
  \cline{3-11}
          &       &shots&sliver ($\mu$m)&CPU(s)    &shots&sliver ($\mu$m)&CPU(s)    &shots&sliver ($\mu$m)&CPU(s) \\
  \hline
  C432    &1109	  &6898	    &48.3    &8.51		   &4371   &23.5  &10.0  	      &4214	     &7.4	 &1.87  \\ 
  C499    &2216	  &13397	&96.0    &16.9		   &8325   &45.0  &19.5  	      &8112	     &11.8	 &2.6   \\ 
  C880    &2411	  &17586	&160.5	 &24.93		   &11020  &84.4  &29.7  	      &10653	 &28.6	 &3.8   \\ 
  C1355   &3262	  &23283	&185.2	 &29.44		   &14555  &87.8  &33.6  	      &13936	 &24.8	 &5.1   \\ 
  C1908   &5125	  &35657	&333.6	 &48.78		   &22352  &181.1 &57.3           &21540	 &88.0	 &7.68  \\ 
  C2670   &7933	  &56619	&525.4	 &84.11		   &35424  &274.4 &96.9  	      &34102	 &114.8	 &11.89 \\ 
  C3540   &10189  &74632	&668.5	 &114.33	   &46617  &360.0 &133.7 	      &44901	 &129.8	 &15.98 \\ 
  C5315   &14603  &108761	&950.4	 &176.89	   &67795  &488.9 &200.8          &65222	 &190.2	 &23.85 \\ 
  C6288   &14575  &103148	&819.2	 &175.65	   &64987  &382.0 &201.0          &62416	 &86.1	 &22.64 \\ 
  C7552   &21253  &151643	&1334.6	 &242.77	   &94902  &717.6 &280.0          &91157	 &290.7	 &32.02 \\ 
  S1488   &4611	  &37126	&303.7	 &55.03		   &22984  &146.2 &64.7           &22099	 &31.6	 &8.14  \\  
  S38417  &67696  &454307   &4040.2  &727.1        &285049 &2293.0&1020           &275054    &729    &88.5  \\
  S35932  &26267  &163956   &1470.4  &228.02       &103960 &808.3 &256.4          &100629    &284    &34.85 \\
  S38584  &168319 &1096363  &10045.2 &2268.6       &690054 &5777.0&3565.2         &666906    &1801.7 &216.39\\
  S15850  &34660  &231681   &2012.8  &329.99       &145745 &1085.1&414.6          &140879    &320    &44.7  \\
  \hline
  avg.    & -     &171670   &1533.0  &302.1        &107876 &850.3 &425.6          &104121    &275.9  &34.7  \\
  ratio   & -     &1        &1       &1            &0.63   &0.55  &1.41           &\textbf{0.61}&\textbf{0.18}&\textbf{0.11}\\
  \hline \hline
\end{tabular}
\end{table*}
%}}}

%To demonstrate the effectiveness of our algorithm, we compare it with the work of \cite{EBL_SPIE2011_Ma} and the RM method. 
Table \ref{tab:compare} shows the results of our DLF algorithm in comparison with the approaches in \cite{EBL_SPIE2011_Ma} and the RM algorithm.
Since the framework \cite{EBL_SPIE2011_Ma} is adopted to provide the input rectangles, the RM algorithm is denoted as ``\cite{EBL_SPIE2011_Ma}+RM''.
Column ``poly\#'' lists the number of polygons of each test circuit.
All fracturing methods are evaluated with the sliver parameter $\epsilon = 5nm$.
For each method, columns ``shots'', ``sliver'', and ``CPU'' denote the shot number, total sliver length, and runtime, respectively.
First we compare the fracturing algorithm in \cite{EBL_SPIE2011_Ma} and the RM algorithm.
From the table we can see that as an incremental algorithm, the RM algorithm can further reduce the shot number by 37\%, and the sliver length by 45\%.
Meanwhile, the runtime increasing is reasonable: RM algorithm introduces 41\% more runtime.
Besides, we compare our DLF algorithm with other two methods.
We can see that DLF demonstrates the best performance, in terms of both runtime and performance.
Compared with traditional sliver aware rectangular fracturing \cite{EBL_SPIE2011_Ma}, it can achieve around 9$\times$ speed-up.
Besides, the shot number and the sliver length can be significantly reduced (39\% and 82\%, respectively).
Even compared with RM algorithm, DLF is better in terms of performance: it can further reduce the shot number and the sliver length by 3.2\% and 67\%, respectively.

% maybe you can draw a figure to show the trend of CPU time versus number of polygons? 
\begin{figure}[bht]
  \centering
  \vspace{-.1in}
  \includegraphics[width=0.44\textwidth]{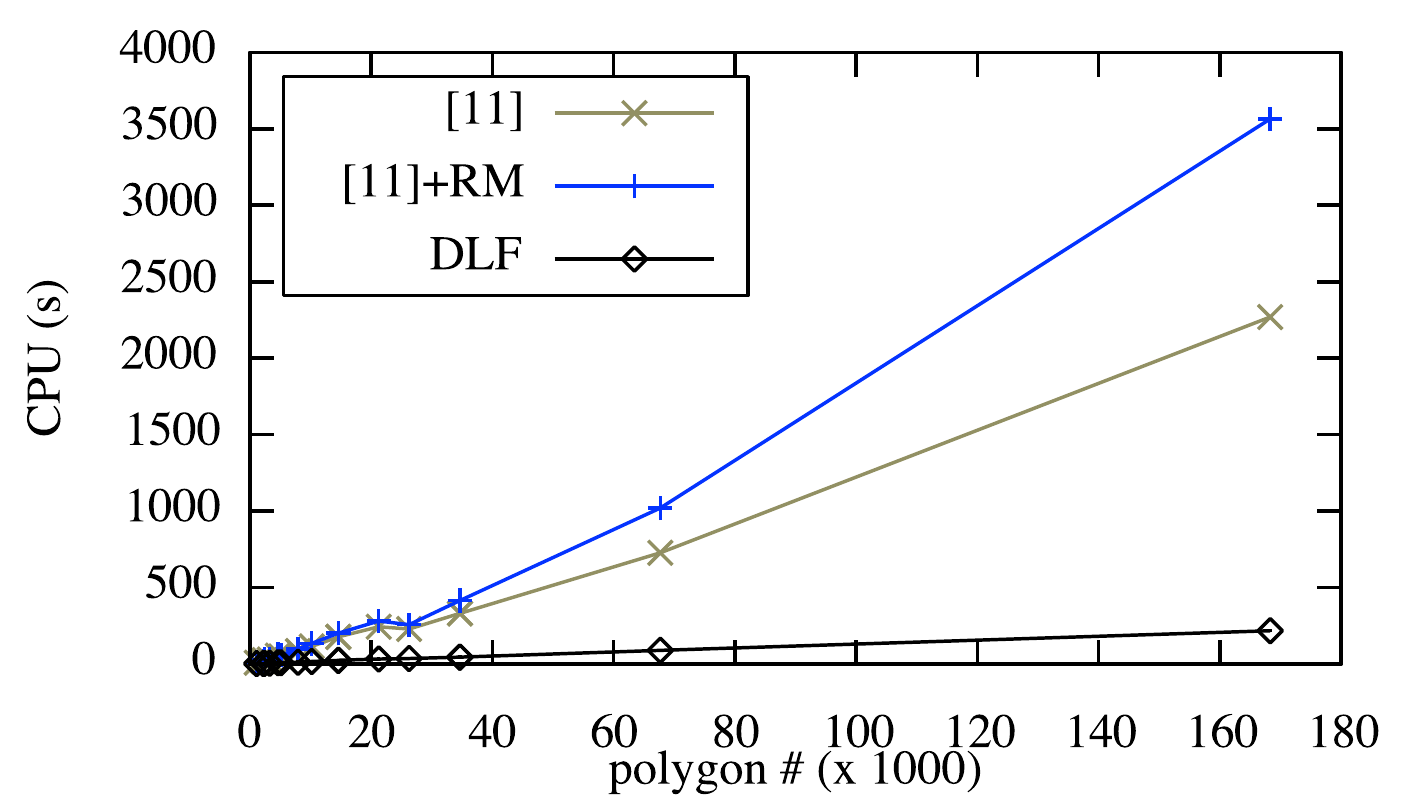}
  \nocaptionrule
  \caption{Comparison on algorithm scalability.}
  \label{fig:scalability}
  \vspace{-.1in}
\end{figure}

In order to evaluate the scalability of our algorithm, we summarize all the run time from Table \ref{tab:compare},
and display in Fig. \ref{fig:scalability}.
Here the X axis denotes the number of polygons (e.g., the problem size), and the Y axis shows the runtime.
We can see that DLF algorithm scales better than both \cite{EBL_SPIE2011_Ma} and RM algorithm.

\vspace{-.1in}
\section{Conclusions}
\label{sec:conclu}

In this paper we have proposed two novel algorithms for EBL with the new L-shape based layout fracturing for shot number and sliver minimization. 
The rectangular merging (RM) based algorithm is optimal for a given set of rectangular fractures.  However, to get better performance, we show that the direct L-shape fracturing (DLF) algorithm is superior by directly decomposing the original layouts into a set of L-shapes. 
DLF obtained the best results in all metrics, including shot count, sliver, as well as runtime compared to the previous state-of-the-art rectangular fracturing with RM.
To our best knowledge, this is the first systematic and algorithmic effort in EBL L-shaped fracturing with sliver minimization. As EBL is widely used for mask making and also gaining momentum for direct wafer writing, we believe a lot more research can be done, for not only layout fracturing but also EBL-aware physical design. 

% are very promising for both performance improvement and speed-up.
%Since this is the first systematic attempt on L-shape layout fracturing problem, there is still some room to improve.
%In addition, solving this problem could have more benefits as EBL may be also developed as the solution for next generation lithography process.
%It shall be also noted that solving this problem more effectively will have more benefits as EBL may become more and more widely used for next generation lithography process.

\vspace{-.1in}
\section*{Acknowledgment}

This work is supported in part by NSF and NSFC.

{
\vspace{-.1in}
\bibliographystyle{IEEEtran}
\bibliography{Ref/Bei,Ref/Algorithm,Ref/EBL,Ref/Lithography,Ref/DPL}
}

\end{document}